\RequirePackage{amsmath}
\documentclass{article}       
\newif\iffull
\fulltrue
\allowdisplaybreaks

\usepackage{amsfonts}

\usepackage{amsthm}
\usepackage{graphicx}

\usepackage[citecolor=blue, linkcolor=blue, colorlinks=true]{hyperref}

\newcommand\refTheorem[1]{\hyperref[#1]{Theorem~\ref*{#1}}}
\newcommand\refLemma[1]{\hyperref[#1]{Lemma~\ref*{#1}}}
\newcommand\refProposition[1]{\hyperref[#1]{Proposition~\ref*{#1}}}
\newcommand\refCorollary[1]{\hyperref[#1]{Corollary~\ref*{#1}}}
\newcommand\refFigure[1]{\hyperref[#1]{Figure~\ref*{#1}}}
\newcommand\refAlgorithm[1]{\hyperref[#1]{Algorithm~\ref*{#1}}}
\newcommand\refMechanism[1]{\hyperref[#1]{Mechanism~\ref*{#1}}}
\newcommand\refEquation[1]{\hyperref[#1]{(\ref*{#1})}}
\newcommand\refExample[1]{\hyperref[#1]{Example~\ref*{#1}}}
\newcommand\refSection[1]{\hyperref[#1]{Section~\ref*{#1}}}
\newcommand\refSubsection[1]{\hyperref[#1]{Subsection~\ref*{#1}}}
\newcommand\refDefinition[1]{\hyperref[#1]{Definition~\ref*{#1}}}
\newcommand\refRemark[1]{\hyperref[#1]{Remark~\ref*{#1}}}

\theoremstyle{plain}

\theoremstyle{definition}

\usepackage{times}
\usepackage{xcolor}
\usepackage{soul}
\usepackage[utf8]{inputenc}
\usepackage[small]{caption}
\newcommand{\citey}[1]{\cite{#1}}
\usepackage{enumitem}

\newtheorem{theorem}{Theorem}
\newtheorem{definition}{Definition}
\newtheorem{corollary}{Corollary}
\newtheorem{proposition}{Proposition}
\newtheorem{lemma}{Lemma}
\theoremstyle{remark}

\newtheorem{example}{Example}

\usepackage{latexsym} 

\title{Facility Reallocation on the Line}
\author{Bart de Keijzer \\ \small{\texttt{bart.de\_keijzer@kcl.ac.uk}}  \\ \small{King's College London} \and Dominik Wojtczak \\ \small{\texttt{d.wojtczak@liverpool.ac.uk}} \\ \small{University of Liverpool}}
\begin{document}

\maketitle

\begin{abstract}
We consider a multi-stage facility reallocation problems on the real line, where a facility is being moved between time stages based on the locations reported by $n$ agents. The aim of the reallocation algorithm is to minimise the social cost, i.e., the sum over the total distance between the facility and all agents at all stages, plus the cost incurred for moving the facility.
We study this problem both in the offline setting and online setting. In the offline case the algorithm has full knowledge of the agent locations in all future stages, and in the online setting the algorithm does not know these future locations and must decide the location of the facility on a stage-per-stage basis. 
We derive the optimal algorithm in both cases. For the online setting we show that its competitive ratio is $(n+2)/(n+1)$.
As neither of these algorithms turns out to yield a strategy-proof mechanism, we propose another strategy-proof mechanism which has a competitive ratio of $(n+3)/(n+1)$ for odd $n$ and $(n+4)/n$ for even $n$, which we conjecture to be the best possible.
We also consider a generalisation with multiple facilities and weighted agents, for which we show that the optimum can be computed in polynomial time for a fixed number of facilities.
\end{abstract}

\section{Introduction}
Facility location is one of the most well-studied problems in the literature due to its multitude of practical applications, e.g., to clustering of images \cite{songCVPR17}, and to document and image summarisation \cite{lin2012learning,tschiatschek2014learning}.
In its simplest form, also referred to as the Weber problem \cite{weber1909standort}, the aim is to locate a single point from which the sum of the transportation costs to $n$ agents' locations is minimal. The generalisation of this problem where the task is to place $k$ facilities in a way that the sum of the distances of each agent to its nearest facility is minimised, is NP-hard already in two-dimensions \cite{megiddo1984complexity}. 
However, it is polynomial time solvable in the one-dimensional setting \cite{megiddo1983maximum}, i.e., when the agents' and facilities' locations are all placed along a single real line. Such scenarios were studied, e.g., in the context of an optimal placement of public facilities along a street \cite{miyagawa2001locating} or to analyse voting scenarios \cite{feldman2016voting}.

We generalise this classic facility location problem to the situation where the interaction between agents and facilities lasts over multiple rounds, the agents' locations may not be known in advance and the facilities can be moved if needed. In particular, let us consider the following motivating example. Assume there is a political party with $k$ members that would like to win the next $T$ consecutive parliamentary elections. In order to achieve this, the party would like its members to represent the political opinions of as many voters as possible to get their votes. A voter feels well-represented if at least one party member has a similar political stance as her. As a result, a party that would like to succeed should try to gather members with a diverse range of political opinions.\footnote{In reality, it would not be possible to choose such political positions in a completely arbitrary way, as a certain degree of consistency in a party's program is needed for it to be taken seriously by the public. This raises some interesting open research questions.} During each term, the political opinion of the voters may change and the party may need to refocus and reconsider their positions, to better reflect current political sentiments. At the same time, each time a politician changes their opinion, they lose a bit of credibility. To estimate such a difference in opinions, \citey{downs1957economic} proposed to model the political views as a spectrum, ranging from extreme-left to extreme-right, as points along as a single real line. The ultimate aim for the party is then to minimise the sum of the distances from its voters while simultaneously taking into consideration the credibility that is lost when readjusting the party's political stance before each election.

In an alternative formulation, one can imagine a long and narrow beach where $k$ ice-cream vendors (owned by the same company) are to be located. For the next $T$ hours, the beach is visited by $n$ customers and their location may change throughout the day. As each client will typically simply pick the closest vendor, it is best for the vendors to change their location throughout the day to adjust to the demand. The aim in this case is the minimisation of the social cost, i.e., the total distance that the customers as well as the ice-cream vendors have to travel.


The models we described so far assumed the agents to report their location truthfully. 
However, since each agent would like to be as close as possible to one of the facilities, 
they may have an incentive to lie, and misreport their location as an attempt to make the facility move closer to their real location. From the point of view of the facility owner, such untruthful behaviour is highly undesirable, as the reported information needs to be reliable for making effective decisions on relocating the facilities. Thus, one typically strives to devise a {\em strategy-proof mechanism}, where the term {\em mechanism} simply refers to an algorithm that takes inputs from multiple independently acting self-interested agents, and outputs a facility assignment based on the locations reported by these agents, while {\em strategy-proof} refers to the property that under this mechanism no agent can gain by misreporting their location.
There are various very important mechanism design domains where strategy-proofness is attained by allowing the mechanism to charge a payment from the agents, where the \emph{utility function} that an agent is trying to maximise is then modeled by including the payment as a negative term. Examples of such domains include many auction scenarios, where an auctioneer runs a mechanism to sell one or more items, and the participating agents can receive these items in exchange for a payment.
However, there are also many domains where payments are impossible or undesirable, 
e.g., in kidney exchanges, public projects, politics, and voting settings. This impossibility may arise due to e.g. ethical, legal or privacy issues. 
In such settings, the mechanism proposed needs to be strategy-proof, without using any monetary transfers. One of the aims (among others) of the present paper is to design strategy-proof mechanism without money for the facility reallocation problem.
%

\paragraph{Outline of this paper.}
Our analysis starts in Section \ref{sec:regular} with finding an optimal algorithm in the case the true locations of all the agents are known, which we call the {\em offline} setting. We show that there is an algorithm for the offline setting that runs in linear time for one facility ($k=1$) and another one that runs in polynomial time for any fixed $k$.
We then adapt our algorithm for $k=1$ to the {\em online} setting in Section \ref{sec:online}.
In such a setting, at each time stage we are required to make the decisions on the location of the subsequent time stage, before seeing the remainder of the input (i.e., the locations of the agents in future stages), which makes it impossible to find a solution of the same quality as the optimal offline solution. However, for the online setting we are able to minimise the {\em competitive ratio} instead, which is the worst-case ratio of the cost returned by the online algorithm and the optimal offline cost. We show a mechanism of which the competitive ratio is $(n+2)/(n+1)$ and prove that no other algorithm can do better.
Finally, in Section \ref{sec:sp}, we show that neither of these one facility location algorithms yields a strategy-proof mechanism, and we devise a new strategy-proof mechanism without monetary transfers. We show that the competitive ratio of this mechanism is $(n+4)/n$ for odd $n$ and $(n+3)/(n+1)$ for even $n$, and that these values are tight.

A preliminary version of this paper, where most of the proofs were omitted has appeared as \citey{de2018facility}.


\subsection{Related Work}\label{sec:related}
The body of literature on facility location is extensive and very diverse in the large array of variations of the problem that has been considered in past literature. We limit our discussion in this setting to the papers that are, to the best of our knowledge, most closely related to ours.

Since an earlier publication of a preliminary conference version of the present paper \citey{de2018facility}, direct follow up work to has appeared in \citey{fotakisetal} 
where the authors present a polynomial time algorithm for the generalisation of the reallocation problem where there are multiple (i.e., $K \geq 1$) facilities. The main results of \citey{fotakisetal} are an algorithm for computing the optimal solution in the offline variant of the problem, where all agent locations at all stages are known in advance. This algorithm runs in time polynomial in $n$, $T$, and $K$. Additionally, the authors present an online algorithm with an analysis that bounds its competitive ratio. In the present paper, we also present an (offline) algorithm for the $K$-facility variant, under the additional generalisation that the objective is to minimise a \emph{weighted} sum of costs of the players. Our algorithm, runs in time exponential in $K$, but polynomial in $n$ and $T$, yielding a polynomial time algorithm for each fixed choice of $K$.

Our work fits tightly into the literature of time-evolving optimisation problems, where an instance of a computational problem changes over time and there is a cost incurred by implementing a change in the solution at each time step. 

See, for example, \cite{evolving1,evolving2,evolving3}, where the latter two works consider two other variants of time-evolving facility reallocation problems.

A mobile facility location problem, which can be seen as a one-stage version of our problem with $k$ facilities, was introduced in \cite{friggstad2011minimizing} where it was shown that this problem is NP-hard in general. 
A polynomial $(3+\epsilon)$-approximation algorithm was given in \cite{ahmadian2013local}. 

The study of the $k$-facility location problem in an online setting, also called the $k$-median problem in such a context, has been extensively studied (see, e.g., \cite{fotakis2011online} for a survey). In particular, \citey{diveki2011online} studied an online model where the location of the facilities can be moved, but with a zero cost.

The papers \citey{duan2019heterogeneous,xu2021two} comprise a recent study on strategy-proofness for a facility location problem on a line, where there are two facilities to be placed, and agents aim to minimise their total distance to both these facilities. Another recent work is \citey{Chen2020185}, where there are again two facilities, and the utility functions of the agents are heterogeneous, where an agent may either maximise or minimise over the distances to the facilities.

The field of approximate mechanism design without money was initiated by \citey{procaccia2009approximate} where the facility location problem was considered. This research has attracted much attention in recent AI conferences. For example, \citey{todo2011false} study false-name strategy-proof mechanisms on a real line, i.e., such mechanisms cannot be manipulated to their advantage by agents who replicate themselves. The paper \citey{sui2013analysis} study strategy-proof facility location in multi-dimensional space for different metrics and devise the {\em percentile mechanisms} for them. 
In \citey{zou2015facility}, strategy-proof mechanisms are studied for agents with dual preferences where some agents would like to be as close as possible to a facility, while others would prefer to be as far as possible.
Moreover, \citey{serafino2015truthful} study the two facility problem where the cost function may differ between agents. 
The paper \citey{filos2017facility} studies strategy-proof mechanisms for double-peaked preferences, which can model e.g., a scenario where each agent would like to be close to a facility, but not too close.
In \citey{procacciaapproximation} the trade-off is studied between variance and approximation factor for strategy-proof mechanisms.
The one-stage facility location problem in the context of voting under the constraint that
the facilities can only be placed on agents' locations is studied in \citey{feldman2016voting}.
Furthermore, \cite{fotakis2014power} characterised completely the deterministic strategy-proof mechanisms for the placement of two facilities on the line and showed that the best approximation ratio of such a mechanism is $n-2$. 
Lastly, \citey{lu2010asymptotically} showed there exists a 4-approximation randomised mechanism for the same problem, while a 1.045 lower bound is also known \cite{lu2009tighter}.

\section{Preliminaries}
For $a \in \mathbb{N}$, we will write $[a]$ to denote the set $\{1,\ldots, a\}$. In this paper we will treat all sets as multisets, and all the operations are thus multiset operators.

An instance of the \emph{facility reallocation problem} is a quadruple $(n,T,y^0,x)$, where $n \in \mathbb{N}$ is the number of agents, $T$ is the number of stages, $y^0$ is the starting location of facility, and $x = (x^1,\ldots, x^T)$ are the vectors of agent locations in each stage, where $x^t = (x_1^t \ldots, x_n^t) \in \mathbb{R}^n$ are the locations of the agents at Stage $t \in [T]$. 
A \emph{solution} of a given instance is a placement of the facility at each of the stages, i.e., a sequence $y = (y^1,\ldots y^T) \in \mathbb{R}^T$. A \emph{mechanism} is a mapping from instances to solutions.\footnote{For convenience, we conflate the terms \emph{algorithm} and \emph{mechanism} from this point.} The \emph{cost} of a solution $y$ is given by 
\begin{equation*}
C(y) = \sum_{t = 1}^T \left(|y_j^{t-1} - y_j^t| + \sum_{i = 1}^n |x_i^t - y_j^t|\right),
\end{equation*}
which is, in words, the sum of distances from each agent to the facility at each stage $t$, plus the total distance the facility moves across all stages. An \emph{optimal solution} is a solution that minimises $C$.
For convenience we denote the individual terms in the above summation by $C^1,\ldots C^T$. So, for $t \in [T]$ we let
\begin{equation*}
C^t(y) = \left(|y_j^{t-1} - y_j^t| + \sum_{i = 1}^n |x_i^t - y_j^t|\right),
\end{equation*}
so that
\begin{equation*}
C(y) = \sum_{t = 1}^T C^t(y). 
\end{equation*}
As $C^t$ is only dependent on the values of $y$ at coordinates $t-1$ and $t$, we may overload notation and occasionally write $C^t(y^{t-1},y^t)$ instead of $C(y)$.

We define $X^t$ as the multiset $\{x^t_1,\ldots, x_n^t\}$. Let $t \in [T]$ be a stage, and let $y^{t-1}$ be any location. We define $M^{t}(y^{t-1})$ as the \emph{median} of the set of points $X^t \cup \{y^{t-1}\}$, i.e., the set of points $z$ such that $\sum_{i = 1}^n |x_i^t - z| + |y^{t-1} - z|$ is minimised. Note that $M^{t}(y^{t-1})$ implicitly depends on the set $X^t$ which is part of a facility reallocation problem instance, but this set $X^t$ will be clear from context at all times throughout our discussion. It is straighforward to verify that $M^t(y^{t-1})$ is the middle point of $\{y^{t-1}\} \cup X^t$ if $n$ is even, and is the interval between (and including) the two middle points of $\{y^t\} \cup X^t$ if $n$ is odd. 

In Section \ref{sec:sp}, we study the \emph{strategy-proofness} property of our mechanisms. There, we assume that the input to the mechanism is provided by the agents, who are interested in minimising their total distance to the facility. They may thus misreport their true locations, in case this results in facility placements closer to their true locations.

Let $A$ be a mechanism. We define the \emph{cost of Agent $i \in [n]$ for a solution $y$} as 
\begin{equation*}
c_i(y) = \sum_{t=1}^T |y^i - x_i^{t}|.
\end{equation*}
We use the notation $(\tilde{x}_S,x_{-S})$ to denote a solution obtained from $x$ by replacing the location vectors $\{x_i : i \in S\}$ where $x_i = (x_i^1,\ldots,x_i^T)$, by different vectors $\tilde{x}_S = \{\tilde{x}_i : i \in S\}$, where $\tilde{x}_i = (\tilde{x}_i^1, \ldots, \tilde{x}_i^T)$ are the alternative locations corresponding to Agent $i \in S$. Mechanism $A$ is \emph{group-strategy-proof} if for all $S \subseteq [n]$, for all $\tilde{x}_S$, there exists an $i \in S$ such that $c_i(A(x)) \leq c_i(A(\tilde{x}_S,x_{-S}))$. Mechanism $A$ is \emph{strategy-proof} if for all $i$ and for all $\tilde{x}_i$ it holds that $c_i(A(x)) \leq c_i(A(\tilde{x}_{i},x_{-i}))$. Thus, stated more informally, strategy-proofness is a property that requires that no agent can improve their cost through reporting a set of locations other than their true locations. Similarly, group-strategy-proofness requires that no set of agents can collectively report alternative locations such that all agents in the set strictly improve their cost.

\section{Optimal Mechanisms}
\label{sec:regular}
First, we consider the basic problem of computing an optimal solution to the facility reallocation problem when the complete instance is given to the mechanism in advance.

Let $I = (n,T,y^0,x)$ be a facility reallocation instance.
The following lemmas show that in every Stage $t \in [T]$, putting the facility on a point in the interval $M^t(y^{t-1})$ is less expensive than putting the facility outside of $M^t(y^{t-1})$, regardless of the choice of facility locations in all the other stages. 
\begin{lemma}\label{lem:easystuff}
Let $y = (y^1, \ldots, y^T)$ be a solution to $I$ and let $t \in [T]$, and let $d$ be the distance between $y^t$ and the nearest point $z \in M^t(y^{t-1})$. Then,
\begin{equation*}
d \geq C^t(y) - C^t((z,y^{-t})),   
\end{equation*}
where $(z,y^{-t})$ is the vector of facility locations obtained from $y$ by replacing $y^t$ with $z$.
\end{lemma}
\begin{proof}
In case $n$ is even, then $M^t(y^{t-1})$ is a single point, located either at $y^{t-1}$ or at one of the agents. We consider only the latter case, and assume that $M^t(y^{t-1})$ is located at an Agent $\ell$, with location $x_{\ell}^t$. The former case is proved by simply replacing $x_{\ell}^t$ by $y^{t-1}$ in the proof that follows. Consider the list $x^{\uparrow}$ in which the multiset of points $(\{y^{i-1}\} \cup X^t) \setminus \{x_\ell^t\}$ is ordered non-decreasingly. Note that $x^{\uparrow}$ consists of $n$ entries. The cost $C^t((w,y^{-t}))$ of placing the facility at any point $w$ can now be written as:
\begin{equation*}
\sum_{i=1}^{n/2} (|x^{\uparrow}_i - w| + |x^{\uparrow}_{n-i+1} - w|) + |x_{\ell}^t - w|.
\end{equation*}
Note that the $i$th term in the above summation is at least $|x^{\uparrow}_i - x^{\uparrow}_{n-i+1}|$, and in case $w$ lies in between $x^{\uparrow}_i$ and $x^{\uparrow}_{n - i + 1}$ then this holds with equality. Moreover, it is straighforward to verify that in case $w$ lies at a distance $c$ of the interval $[x^{\uparrow}_i,x^{\uparrow}_{n-i+1}]$, then the $i$th term in the summation is equal to $2c + |x^{\uparrow}_i - x^{\uparrow}_{n-i+1}|$. Point $z$ lies at distance $0$ of $x_{\ell}^t$ and is in all intervals $[x^{\uparrow}_i,x^{\uparrow}_{n-i+1}], i \in [n/2]$, so point $z$ lies at distance $0$ from all these intervals. Thus,
\begin{equation*}
C^t(z,y^{-t}) = \sum_{i=1}^{n/2} (|x^{\uparrow}_i - x^{\uparrow}_{n-i+1}|),
\end{equation*}
and $z$ minimises the total cost at Stage $t$, given $y^{-t}$. Point $y^t$ lies at distance $d$ from point $z = x_{\ell}^t$, Hence, 
\begin{equation*}
C^t(y) \geq \sum_{i=1}^{n/2} (|x^{\uparrow}_i - x^{\uparrow}_{n-i+1}|) + d = C^t((z,y^{-t})) + d,
\end{equation*}
which proves the claim for even $n$. 

In case $n$ is odd, define $x^{\uparrow}$ now as the list in which the multiset of points $\{y^{t-1}\} \cup X^t$ is ordered non-decreasingly. Note that $x^{\uparrow}$ consists of an even number of $n+1$ entries. The cost of placing the facility at any point $w$ can now be written as
\begin{equation*}
\sum_{i=1}^{(n+1)/2} (|x^{\uparrow}_i - w| + |x^{\uparrow}_{n-i+1} - w|).
\end{equation*}
In case $w$ lies at a distance $c$ of the interval $I_i = [x^{\uparrow}_i,x^{\uparrow}_{n-i+1}]$, then the $i$th term in the summation is equal to $2c + |x^{\uparrow}_i - x^{\uparrow}_{n-i+1}|$. 
The point $y^t$ lies at distance $d$ from $M^t(y^{t-1}) = [x^{\uparrow}_{\lfloor n/2 \rfloor}, x^{\uparrow}_{\lfloor n/2 \rfloor + 1}] = I_{(n+1)/2}$, the interval corresponding to the last term of the above summation. The point $z$ lies in all intervals $I_1, \ldots, I_{(n+1)/2}$ of the above summation, and this establishes the claim for odd $n$.
\end{proof}
The following lemma is proved using the former.
\begin{lemma}\label{lem:inthemedian}
Let $y = (y^1, \ldots, y^T)$ be a solution to $I$. Suppose that there is a Stage $t$ such that $y^t$ is not in $M^t(y^{t-1})$. Then, replacing $y^t$ with the nearest point $\tilde{y}^t$ to $y^t$ that lies in $M^t(y^{t-1})$ results in a solution with a cost that is at most $C(y)$.
\end{lemma}
\begin{proof}
    We can write the difference in costs of $y$ and $(\tilde{y}^t, y^{-t})$ as follows:
    \begin{align*}
        & C(y) - C((\tilde{y}^t, y^{-t})) \\
        & \qquad = \sum_{u = 1}^T (C^u(y) - C^u((\tilde{y}^t, y^{-t}))) \\
        & \qquad = C^t(y) - C^t((\tilde{y}^t, y^{-t})) + C^{t+1}(y) - C^{t+1}((\tilde{y}^t, y^{-t})) \\
        & \qquad = \sum_{i = 1}^n (|x_i^t - y^t| - |x_i^t - \tilde{y}^t|) + |y^{t-1} - y^t| - |y^{t-1} - \tilde{y}^t| \\
        & \qquad \qquad + \sum_{i = 1}^n (|x_i^{t+1} - y^{t+1}| - |x_i^{t+1} - y^{t+1}|) + |y^{t} - y^{t+1}| - |\tilde{y}^{t} - y^{t+1}| \\
        & \qquad = \sum_{i = 1}^n (|x_i^t - y^t| - |x_i^t - \tilde{y}^t|) + |y^{t-1} - y^t| - |y^{t-1} - \tilde{y}^t| \\
        & \qquad \qquad + |y^t - y^{t+1}| - |\tilde{y}^t - y^{t+1}| \\
        & \qquad = C^t(y) - C^t((\tilde{y}^t, y^{-t})) + |y^t - y^{t+1}| - |\tilde{y}^t - y^{t+1}| \\
        & \qquad \geq |y^t - \tilde{y}^t| + |y^{t} - y^{t+1}| - |\tilde{y}^{t} - y^{t+1}| \\
        & \qquad \geq |y^t - \tilde{y}^t| - |y^t - \tilde{y}^t| \\
        & \qquad = 0.
    \end{align*}
    where in the second-to-last inequality we used Lemma \ref{lem:easystuff}, where in this case $d = |y^t - \tilde{y}^{t}|$. 
\end{proof}

Lemma \ref{lem:inthemedian} yields an easy and efficiently computable optimal mechanism when $n$ is even: An optimal facility reallocation mechanism for $k=1$ always places the facility at Stage $t \in [T]$ in the median interval $M^t(y^{t-1})$. Hence, when the number of agents is even, the optimal allocation vector is unique and can be computed in $O(nT)$ (i.e., linear) time.

For $n$ odd, the above does not yet characterise the optimal mechanism, and it turns out that in this case the facility cannot be placed at just \emph{any} point in the median without sacrificing optimality. This is due to the fact that the median $M^t(y^{t-1})$ of Stage $t$ is dependent on the location $y^{t-1}$ of the facility of the previous stage, and is therefore by recursion also dependent on the location the facility and all the agents at all previous stages. Because $M^t(y^{t-1})$ is generally an interval of points instead of a single point, there is a choice to be made that influences the medians of all the subsequent stages. 

The following two example instances show that the optimal choice of facility at a given stage may depend on the locations of the agents in the next stage.
\begin{figure}[t]
	\begin{center}
		\includegraphics[scale=1]{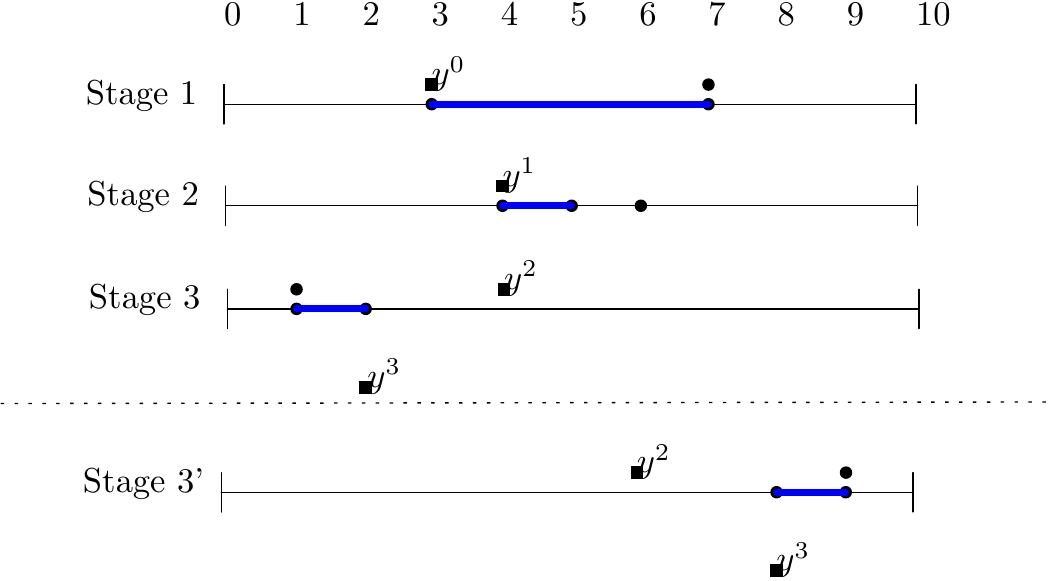}\caption{Depiction of the two facility reallocation instances of Example \ref{exa:1}, one consisting of Stages 1,2,3, and the other consisting of stages 1,2,$3'$. The dots indicate the locations of the agents at each stage. The squares indicate an optimal choice of facility locations, where the square at a given stage is the facility location at the previous stage. (At the first stage it is the starting location.) The square below the final stage is the facility location at the last stage. The blue part of the line at Stage $t$ represents the median $M^t(y^{t-1})$. The bottom part of the figure illustrates that in case Stage 3 of this instance would be replaced with Stage 3', then the facility placement at stages two and 3 would need to be chosen diffently than the solution presented in this figure for the original instance. }\label{fig:1}
	\end{center}
\end{figure}

\begin{example}\label{exa:1}
Consider first the following example with $T=3$ stages and $n=3$ agents, depicted in Figure \ref{fig:1}. Let $y^0 = 3$ be the initial facility location. The locations of the agents at each of the $3$ stages are $x^1 = (3,7,7), x^2 = (4,5,6), x^3 = (1,1,2)$. The median in the first stage is the interval $[3,7]$. The point in this median that we choose for $y^1$ influences the median in the second stage:
\begin{itemize}
	\item When we set $y^1 \in [3,4]$, the median in the second stage will be $[4,5]$;
	\item When $y^1 \in (4,5]$, the median in the second stage will be $[y^1,5]$;
	\item When $y^1 \in [5,6)$, the median in the second stage will be $[5,y^1)$;
	\item When $y^1 \in [6,7]$, the median in the second stage will be $[5,6]$.	
\end{itemize}

The optimal solution is to set $y^1 \in [4,5]$, and to not move the facility to a different location in the second stage. 
This is the best tradeoff to minimising the second stage's cost while keeping the facility close to the agent locations in the third stage so that the third stage's cost is also kept small. 

However, if in Stage $3$ the facilities of the three agents would be $\bar{x}^3 = (8,9,9)$, then the optimal choice of facility location for the first stage would be to set $y^1 \in [5,6]$. \qed
\end{example}

The above examples show that there may be infinitely many optimal solutions when $n$ is odd. The analysis also suggests that it may always be optimal to put the facility at any given stage at the location of the central agent of the subsequent stage, whenever that is possible. We can prove this, and in fact we can refine this statement further, as follows.

\begin{theorem}
   \label{thm:optimal-position}
Suppose that in instance $I$ it holds that $n$ is odd. There exists an optimal solution $y$ for this instance such that:
\begin{itemize}
	\item at any Stage $t \in [T-1]$, the facility is placed at the point in $M^t(y^{t-1})$ that lies closest to the location of the middle agent at the subsequent Stage $t+1$, i.e., the median of $\{x_1^{t+1},\ldots,x_n^{t+1}\}$. 
	\item At Stage $T$, the facility is placed anywhere in $M^t(y^{t-1})$.
\end{itemize}
\end{theorem}
\begin{proof}
    We assume without loss of generality (by possibly renaming the agents at each stage) that $x_i^t \leq x_{i+1}^{t}$ for all $i \in [n-1]$ and all $t \in [T]$ so that $x_{\lceil n/2 \rceil}^t$ is the location of the middle agent for each Stage $t$.
    We also assume without loss of generality that the starting location $y^0$ is located at the right of the middle agent $x_{\lceil n/s \rceil}$. We prove by induction on the number of stages $T$ that the claim holds. We additionally prove at each stage of our inductive proof the following \emph{auxiliary claim}: 
    \begin{itemize}
    	\item Changing the instance $I$ by moving the starting location $y^0$ a distance $d$ further to the right (i.e., away from the middle agent $x_{\lceil n/s \rceil}^1$ at Stage 1) does not decrease the optimal cost, and will increase the optimal cost by at most $d$.
    \end{itemize} 
	
	Our induction basis is when $T=1$, where our claim follows from the definition of $C$: Letting $x^{\uparrow}$ denote the non-decreasingly ordered list $X^1 \cup \{y_0\}$, the distance between the facility and any pair of points $x^{\uparrow}_i, x^{\uparrow}_{n-i+1}$ (where $i \in [(n+1)/2]$) is minimised when $y^1$ is in between these two points, and placing the facility anywhere in $M^1(y^0)$ ensures that the facility is placed in between all these pairs. The cost is then equal to the length of all the intervals. Moreover, changing the instance by moving the starting location $y^0$ an amount of $d$ to the right can only lengthen the set of intervals, and will increase the total length of the intervals by $d$. This shows that the base case holds.
    
    Suppose now that the claim, including the auxiliary claim, holds for all instances with $T = U$ stages. We prove that it also holds when $I$ has $T=U+1$ stages. 
    
    Consider the subinstance of $I$ restricted to stages $[U+1]\setminus\{1\}$ after fixing the facility location at stage one to any location $z^1$. Denote this subinstance by $J(z^1)$. We denote the optimal cost of $J$ by $C^*(J(z^1))$. By the induction hypothesis, we may assume that every solution vector for $J$ satisfying the properties of the claim, is optimal for $J$ (i.e., attains cost $C^*(J(z^1))$). We denote by $s(J(z^1))$ such an optimal solution vector.
    
    The cost $C(y)$ of instance $I$ for a given solution $y$ can now be decomposed as the sum of the cost incurred by Stage $1$ and the cost of $(y_2,\ldots, y_n)$ on subinstance $J(y^1)$. If $y$ is an optimal solution, then taking $y$ and replacing the entries $(y^2, \ldots, y^n)$ by $s(J(y^1))$ is also an optimal solution for $I$. The location $y^1$ in the optimal solution should thus satisfy that $C^1(y^0, y^1) + C^*(J(y^1))$ is minimised. By Lemma \ref{lem:inthemedian}, it holds that $y^1 \in M^1(y^0)$, and by the auxiliary claim of the induction hypothesis, $C^*(J(y^1))$ decreases as $y^1$ gets closer to $x_{\lceil n/2 \rceil}$. Furthermore, the term $C^1(y^0, y^1)$ is constant in its argument $y^1$ on the subdomain $M^1(y^0)$. Altogether this means that $y^1$ is the point in $M^1(y^0)$ closest to $x_{\lceil n/2 \rceil}^1$. This proves that $(y^1, s(J(y^1)))$ is an optimal solution for $I$ that satisfies the properties of the main claim, and concludes the proof of the main claim of the induction step.
    
    What remains is to establish the auxiliary claim. To that end, let $\tilde{I}$ be an instance where the starting location $\tilde{y}^0$ lies at a distance $d$ to the right of $y^0$, and where $\tilde{I}$ is otherwise identical to $I$. Through analogous reasoning as above, we infer that there is a solution $\tilde{y} = (\tilde{y}^1, s(J(\tilde{y}^1)))$ for instance $\tilde{I}$ that satisfies the properties of the main claim. Because $x_{\lceil n/s \rceil}^1 \leq y^0 \leq \tilde{y}^0$, we have the inclusion $M^1(y^0)\subseteq M^1(\tilde{y}^1)$, and these two medians share the same leftmost endpoint. We now distinguish three cases.
    \begin{itemize}
    \item If $x_{\lceil n/2 \rceil}^2 \in M^1(y^0)$, then also $x_{\lceil n/2 \rceil}^2 \in M^1(\tilde{y}^0)$, which means that both $y^1$ and $\tilde{y}^1$ are placed at location $x_{\lceil n/2 \rceil}^2$. This means that in both instances $I$ and $\tilde{I}$ the contribution to the cost by the subinstance $J$ is the same, and the only difference in cost is caused by the cost contribution of Stage 1. The difference in cost contribution of Stage 1 among the two instances is $d$, which follows from analogous reasoning as in the proof of the base case of the induction. This means that the total cost difference between optimal solutions for instances $I$ and $\tilde{I}$ is $d$, and this establishes the auxiliary claim for this case.
    \item If $x_{\lceil n/2 \rceil}^2$ lies strictly to the left of $M^1(y^0)$, which means that also in this case, $y^1$ and $\tilde{y}^1$ are the same location, which is the leftmost endpoint of $M^1(y^0)$ (and $M^1(\tilde{y}^0)$). For reasons analogous to the former case, it follows that the total cost difference between optimal solutions for instances $I$ and $\tilde{I}$ is $d$.
    \item In the last case, $x_{\lceil n/2 \rceil}^2$ lies strictly to the right of $M^1(y^0)$.  Because in the two instances, in the first stage, the facilities are placed at the point in the median closest to $x_{\lceil n/2 \rceil}^2$, it holds for the subinstances $J(y^1)$ and $J(\tilde{y}^1)$ that the starting positions are at distance at most $d$ from each other, and that the starting positions both lie on the same side of the middle agent location $x_{\lceil n/2 \rceil}^2$ at the first stage of these subinstances. Thus, the auxiliary claim of the induction hypothesis applies to the resulting pair of subinstances and we may conclude that in these subinstances the optimal costs differ by at most $d$, where $C^*(J(\tilde{y}^1)) \leq C^*(J(y))$. Next, we observe that the optimal cost of $I$ in the first stage is exactly $d$ less than the optimal cost of $\tilde{I}$ in the first stage, so we conclude that the optimal cost of $I$ is still at most the optimal cost of $\tilde{I}$, and is at most $d$ less. This completes the proof of the auxiliary claim of the induction step.
    \end{itemize} 
\end{proof}

The following corollary summarises all of the above.
\begin{corollary}\label{cor:median2}
It is an optimal facility reallocation mechanism for $k=1$ to place the facility at each stage $t \in [T-1]$ at the point in the median interval $M^t(y^{t-1})$ that lies closest to the middle agent of Stage $t+1$, and to place the facility at Stage $T$ at any point in the median interval. Hence, when the number of agents is even, the optimal allocation vector is unique and can be computed by an online mechanism. When the number of agents is odd, the optimal mechanism needs to look at each stage $t \in [T-1]$ at the agent locations in Stage $t$ and $t+1$ only. The mechanism runs in both cases in $O(Tn)$ time.
\end{corollary}
Thus, for $n$ odd, we can compute the optimum efficiently, but we do need a one stage ``look-ahead''. Thus, this result does not imply an optimal \emph{online} mechanism. We give in Section \ref{sec:online} an online mechanism with an optimal competitive ratio.

\section{The Weighted Problem with Multiple Facilities} 
Next, we consider a generalised variant of the problem where there are $k \geq 1$ facilities and the agents have weights. The cost of an agent $i \in [n]$ is their distance to the nearest facility at each stage, and their weight $w_i \in \mathbb{R}_{\geq 0}$ is the factor by which their cost contributes to the cost function.

The problem of computing the optimal facility locations for such a generalised instance is considerably more complex. We prove that nonetheless, when the number of facilities $k$ is fixed, this can be done in polynomial time, when $k$ is fixed.
\begin{theorem}\label{thm:polytimeoffline}
There exists a mechanism that computes the optimal solution to a generalised facility reallocation problem in time $O(T^2(2\max\{Tn,k\})^{k+1})$.
\end{theorem}
As mentioned in Section \ref{sec:related}, the paper \citey{fotakisetal} provides a strongly related and important result for the special case of this problem where all weights are equal: In this case, the authors show that there exists an algorithm that is not only polynomial in $T$ and $n$, but also polynomial in $k$. Their algorithm can be generalised to handle the case where the objective is a weighted sum of the total movement cost and the total distance of the agents to the facility. However, their algorithm is not applicable to arbitrarily weighted agents, which the algorithm presented here is suitable for. 
\begin{proof}[Proof of Theorem \ref{thm:polytimeoffline}.]
The main insight that we need is that it suffices to consider only solutions where at each stage $t \in [T]$ each facility is placed on a location corresponding to one of the agent locations (at any stage) or to one of the starting facility locations $y_1^0, \ldots, y_k^0$.
\begin{lemma}\label{lem:prevpositions}
Let $y = (y^1, \ldots, y^T)$ be a solution to a generalised facility location instance (where $y^t$ are $k$-dimensional vectors) with $T$ stages. There exists a solution $\tilde{y}$ such that $C(\tilde{y}) \leq C(y)$ and for all $t \in [T]$ and $j \in [k]$, it holds that 	$y_j^t \in \{y_1^0, \ldots y_k^0\} \cup X^1 \cup \cdots \cup X^T$.
\end{lemma}
\begin{proof}
    Given a Stage $t \in [T]$, let $X^{\leq t}$ denote the union of the set of agent positions up to Stage $t$ and and the set of all facility locations up to Stage $t$, i.e.
    \begin{equation*}
    X^{\leq t} = \{x_i^s : i \in [n], s \in [t]\} \cup \{y_1^s, \ldots, y_k^s : s \in [t] \cup \{0\}\}.
    \end{equation*}
    It suffices to show that in case $y$ does not satisfy that $y_j^t \in X^{\leq t}$ for all $t \in [T]$ and $j \in [k]$ then we can change one of the locations $y_j^t \not\in X^{\leq t}$, with $t \in [T]$, $j \in [k]$, to a location $\tilde{y}_j^t$ such that $\tilde{y}_j^t \in X^{\leq t}$ while not increasing the cost.
    
    Consider therefore the highest $t \in [T]$ for which there exists a $j \in [k]$ such that $y_j^t \not\in X^{\leq t}$. We now consider two possible alternative locations for facility $j$ at Stage $t$: the points $p_l$ and $p_r$ in $X^{\leq t}$ closest to $y_j^t$ to the left and right of $y_j^t$ respectively. We prove that moving facility $j$ to one of these points will not decrease the total cost, which suffices to prove the claim.
    
    Moving the facility $j$ at Stage $t$ to a point different from $y_j^t$ may affect the cost contribution coming from $n+2$ sources: the distance to the closest facility of the $n$ agents at Stage $t$, the distance by which the facility moves from Stage $t-1$ to Stage $t$, and the distance by which the facility moves from Stage $t$ to Stage $t+1$.
    
    Let $S_l \subseteq [n]$ be the set of agents to the left of $y_j^t$ for which facility $j$ is the unique closest facility at Stage $t$ under its current location $y_j^t$, and let $W_l$ be the total weight of these agents. Likewise, Let $S_r \subseteq [n]$ be the set of agents to the right of $y_j^t$ for which facility $j$ is the unique closest facility at Stage $t$ under its current location $y_j^t$, and let $W_r$ be the total weight of these agents. Moving the facility to the left of its current location $y_j^t$ will increase the distance of the facility to the agents of $S_r$ at rate $W_r$ and decrease the distance to the agents of $S_l$ at rate $W_l$. Moreover, doing so will increase the distance to $y_j^{t-1}$ if $y_j^{t-1} < y_j^t$ and decrease the distance to $y_j^{t-1}$ otherwise. The same holds for the distance to $y_j^{t+1}$ (if $t+1 \in [T]$). A symmetric observation holds for moving the facility to the right of $y_j^t$. 
    
    Therefore, if $t+1 \in [T]$ and $y_j^{t+1} \not= y_j^t$, moving the facility to point $p_l$ will not increase the cost if
    \begin{equation*}
    W_l + \mathbf{1}[y_j^{t-1} < y_j^t] + \mathbf{1}[y_j^t < y_j^{t+1}] \leq W_r + \mathbf{1}[y_j^{t-1} \geq y_j^t] + \mathbf{1}[y_j^t \geq y_j^{t+1}],
    \end{equation*}  
    where $\mathbf{1}[\cdot]$ denotes the indicator function that maps to $1$ if the provided argument holds, and maps to $0$ otherwise.
    and moving the facility to point $p_r$ will not increase the cost if the above inequality holds in the opposite direction. This is the case because:
    \begin{itemize}
        \item The difference between the right and left hand side of the inequality is the rate of change in the cost function when moving the facility to the right at point $y_j^t$, i.e., this difference is the partial derivative of $C$ with respect to $y_j^t$. 
        \item The partial derivative of $C$ with respect to $y_j^t$ is monotone in the interval $[p_l,p_r]$, because moving the facility in the direction of $p_l$ can only cause agents left of $p_l$ to possibly join facility $j$ and can only cause any Agent $i$ to the right of $p_l$ to possibly drop facility $j$ and to join an alternative facility that is closer to $i$. Thus, when $C$ is viewed as a function of $y_j^t$ restricted to the interval $[p_l,p_r]$ with the remaining coordinates fixed, then $C$ is maximised at $p_l$ or $p_r$.
    \end{itemize} 
    
    If, on the other hand, $t+1 \in [T]$ and $y_j^{t+1} = y_j^t$, as the rates of change in the cost function when moving left and right respectively, are slightly different (because in this case $y_j^{t+1} \not \in X^t$). In this case, moving the facility to point $p_l$ will not increase the cost if
    \begin{equation*}
    W_l + \mathbf{1}[y_j^{t-1} < y_j^t] - 1 \leq W_r + \mathbf{1}[y_j^{t-1} \geq y_j^t] - 1,
    \end{equation*}  
    and moving the facility to point $p_r$ will not increase the cost if the above inequality holds in the opposite direction. Lastly, if $t = T$, then $y_j^{t+1}$ is not an influencing factor in the rates of change in cost as a consequence of moving the facility to the left and right respectively. Therefore, moving the facility to point $p_l$ will not increase the cost if
    \begin{equation*}
    W_l + \mathbf{1}[y_j^{t-1} < y_j^t] \leq W_r + \mathbf{1}[y_j^{t-1} \geq y_j^t],
    \end{equation*}
    and moving the facility to point $p_r$ will not increase the cost if the above inequality holds in the opposite direction. This condition is equivalent to the previous case.
\end{proof}



Using this lemma, a polynomial time algorithm with the claimed runtime can be constructed through standard dynamic programming techniques: For each possible vector $\bar{y}$ of starting facility locations that are in the set expressed in Lemma \ref{lem:prevpositions} (there are at most $k + (Tn)^k$ of them by the above lemma), we can efficiently find a solution to a subinstance $I$ on stages $t, \ldots, T$ with facility starting positions $\bar{y}$, by considering the optimal solutions to the subinstaces on stages $t+1, \ldots, T$ with all possible starting positions, and using the one that minimises the cost for $I$. 

We now provide the details of this construction. Let $X$ be the set of locations expressed in Lemma \ref{lem:prevpositions}. We also define the functions 
\begin{equation}\label{eqn:abbrevC}
C^{>t} = \sum_{u=t+1}^T C^t \qquad \text{ and } \qquad C^{\leq t} = \sum_{u=1}^t C^t ,
\end{equation}
representing the total cost contributed by stages after $t$ and up to $t$, respectively. We may abuse notation and use as the argument $y$ provided to $C^{>t}$ only the $T-t+1$ vectors of facility locations for Stages $t, \ldots, T$. Likewise, we may provide to $C^{\leq t}$ only the $t$ vectors of facility locations for Stages $1,\ldots, t$.

Let $Y$ denote the set of all vectors $y^t$ of locations for the $k$ facilities at any Stage $t \in [T]$, such that $y_j^t \in X$ for all $j \in [k]$ and $y_j^t \leq y_{j+1}^t$ for $j \in [k-1]$. Observe that there are $\binom{|X|}{k} \leq (nT + k)^k$ such vectors. For $y^t \in Y$, denote by $y^{> t}(y^t) = (y^t, (y^{>t})^{t+1}, \ldots, (y^{>t})^T)$ a sequence of $T-t+1$ facility location vectors that minimises the cost generated by the Stages $t+1, \ldots, T$, given that the vector of facility locations at Stage $t$ is $y^t$. That is, $y^{>t}(y^t)$ minimises the function $C^{>t}$ defined in (\ref{eqn:abbrevC}).

Observe that our mechanism needs to compute $y^{>0}(y^0)$, because $C^{>0} = C$. 
These definitions, together with Lemma \ref{lem:prevpositions}, imply an efficient way to compute for a given $y^t \in Y$ the optimal placement $y^{>t}(y^t)$ of the facilities of the subsequent stages, provided that we have computed the optimal locations $y^{>t+1}(y^{t+1})$ for all $y^{t+1} \in Y^{t+1}$:
\begin{equation*}
y^{>t}(y^t) \in \arg_{(y^{t+1}, y^{>t+1}(y^{t+1}))}\max\{C^{>t}((y^{t+1}, y^{>t+1}(y^{t+1}))) : y^{t+1}  \in Y\}.
\end{equation*}
Lemma \ref{lem:prevpositions} implies that the above expression is true, as it states that indeed it suffices to consider only the vectors $y^{t+1}  \in Y$ in the above max-epression.

Since the summation by which $C^{>t}$ is defined consists of at most $Tn + Tk$ terms that each take $O(k)$ time to compute, and the max-expression above is over a set of $|Y| \leq (nT+k)^k$ values, an appropriate vector $y^{>t}(y^t)$ can be computed in time $O(T(n+k)(Tn+k)^k)$ from the values $y^{>t+1}(y^{t+1}), y^{t+1} \in Y$. Thus, proceeding by standard dynamic programming, it is possible to compute $y^{\geq 0}(y^0)$ in time $O(T(n+k)(Tn+k)^k \cdot T) \subseteq O(T^2(2\max\{Tn,k\})^{k+1})$.
\end{proof}

\section{The Online Setting}\label{sec:online}
In this section we study again the basic facility reallocation problem with a single facility, and we focus on the online variant of the problem, where for each stage $t \in [T]$ the agent locations $x_1^{t+1}, \ldots, x_n^{t+1}$ of the next stage may only be read by the mechanism after the mechanism outputs the facility location $y_1^t$ for the current stage. 

We are interested in finding an online mechanism with an as good as possible \emph{competitive ratio}, which is defined as the ratio of the cost of the solution generated by the online mechanism and the cost of the optimal solution (i.e., the solution generated by the optimal offline mechanism). 

Corollary \ref{cor:median2} points out that for an even number of agents it is optimal to put the facility in the median at each stage, which is a single point. Since this can be done in an online fashion, this mechanism suffices for the case of even $n$, and achieves a competitive ratio of $1$. For odd $n$, 

For odd $n$, Corollary \ref{cor:median2} states that the optimal facility placement at any given stage depends on the location of the middle agent at the subsequent stage, which means that the optimal online mechanism necessarily does not achieve a competitive ratio of $1$. The following example shows that due to the lack of ability to look one stage ahead, any optimal online mechanism cannot achieve a competitive ratio better than $(n+2)/(n+1)$.

\begin{figure}[t]
	\begin{center}
		\includegraphics[scale=0.6]{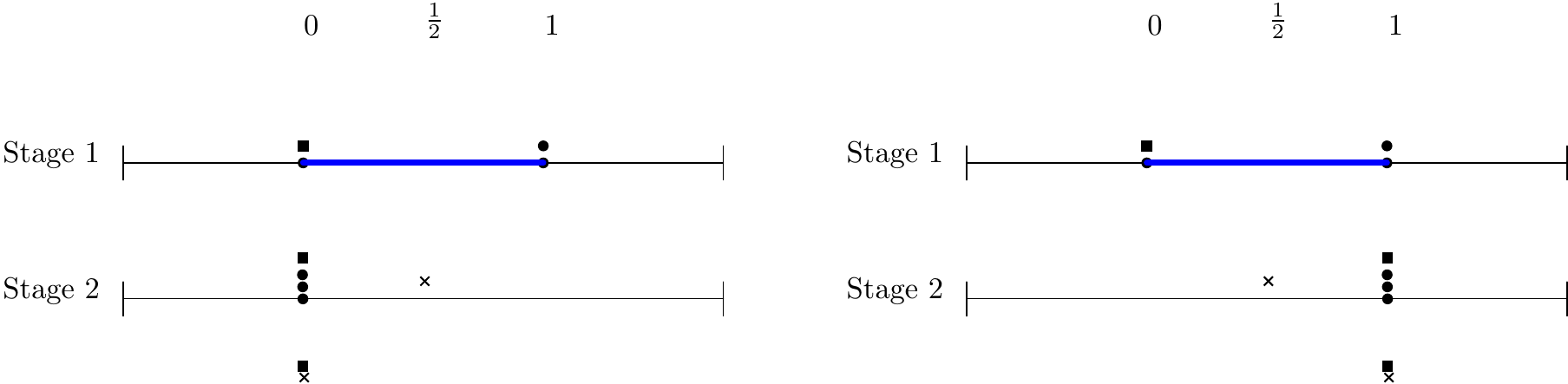}\caption{Depiction of the two facility reallocation instances (one on the left side, and one on the right side) of Example \ref{exa:2}, for the case of $\ell = 1$ (i.e., there are $n=3$ agents). The notation we use is the same as in Figure \ref{fig:1}. The square indicates, at a given stage, the previous stage's optimal location of the facility. The cross indicates, at a given stage, the previous stage's location of the facility that the optimal online mechanism would choose, where the final facility locations are depicted below the final stage. The examples differ in Stage 2, which causes the optimal facility location chosen in Stage 1 to be distinct among the two instances. An optimal online facility location, on the contrary, would choose to set the facility in Stage $1$ at the middlepoint of the median interval $M^1(y^0)$, as this choice minimizes the maximum cost caused by the subseqent stage, of which the online mechanism does not know the precise agent locations when it has to choose the facility location in the first stage.}\label{fig:2}
	\end{center}
\end{figure}

\begin{example}\label{exa:2}
Consider the following two instances $I_{\ell}$ and $I_{\ell}'$, wheree both instances have $2\ell + 1$ agents, for any $\ell \in \mathbb{N}$. Both instances have $T = 2$ stages. Figure \ref{fig:2} depicts the instances for $\ell = 1$. The agent locations of Stage $1$ are $0$ for the first $\ell$ agents and $1$ for the remaining $\ell+1$ agents. At Stage $2$, all agents are located at $0$ in instance $I_{\ell}$, and at $1$ in instance $I_{\ell}'$. The initial facility location is $0$ in both instances. 

The median $M^1(y^0)$ of the first stage is $[0,1]$, so by Corollary \ref{cor:median2} the optimal solution is to place the facility at $0$ in Instance $I_{\ell}$ and at $1$ in Instance $I_{\ell}$. In Stage $2$, the facility then does not need to move.

However, as the instances differ only in the second stage, an online mechanism is restricted to place the facility at the same position in Stage $1$, in both instances. Placing the facility at $1/2$ is the best that any online mechanism can choose, to minimise the maximum cost among those two instances. Therefore the cost of the optimal solution is $\ell+1$ for both instances, while the cost of the solution generated by the optimal online mechanism is $\ell + 3/2$. The ratio of these two quantities is $(n+2)/(n+1)$. This is a lower bound on the competitive ratio achievable by an online mechanism.
\qed
\end{example}

We now provide an online mechanism of which the competitive ratio matches the lower bound on the competitive ratio of Example \ref{exa:2}. The key idea behind this online mechanism is to try to place the facility at each stage as close as possbile to the location where the optimum facility may be placed. While it is impossible to know the exact location of the optimum facility at the current stage $t$, an optimal mechanism can nonetheless derive at each stage the precise interval in which the optimum location may lie: The online mechanism can compute the precise optimum location $\tilde{y}^{t-1}$ at Stage $t-1$ (as it has access to the agent locations of Stage $t$), and by Corollary \ref{cor:median2} the optimal location at Stage $2$ can lie at any point in $M^t(\tilde{y}^{t-1})$, depending on the next stage. Our online mechanism will therefore place the facility $y^t$ at the point $M^t(y^{t-1})$ that lies as close as possible to the middle point of $M^t(\tilde{y}^{t-1})$.

\begin{definition}
	Define Mechanism $A$ as follows: For each Stage $t \in T$, compute the optimum facility location $\tilde{y}^{t-1}$ at Stage $t-1$. Set $y^t$ to the location in $M^t(y^{t-1})$ closest to the middle point of the interval $M^t(\tilde{y}^{t-1})$.
\end{definition}

Regarding the runtime of $A$, note that $\tilde{y}^{t-1}$ can be computed from $\tilde{y^{t-2}}$ in $O(n)$ time, and the point in $M^t(y^{t-1})$ closest to the middle point of $M^t(\tilde{y}^{t-1})$ can be computed in $O(n)$ time as well. We thus obtain the following corollary.
\begin{corollary}
	Mechanism $A$ is an online mechanism that runs in $O(n)$ time per stage, and thus takes $O(Tn)$ (i.e., linear) time in total.
\end{corollary}

A more challenging task is to prove that the competitive ratio of this mechanism is optimal. That is, it matches the lower bound of Example \ref{exa:2}. We will prove this in the remainder of this section.

\begin{theorem}\label{thm:optonline}
	Mechanism $A$ has competitive ratio $(n+2)/(n+1)$ on instances with an odd number of $n$ agents. That is, let $I = (n,T,y^0,x)$ be an instance where $n$ is odd, let $y$ be the output solution of $A$ and let $\tilde{y}$ be the optimal solution. It holds that
	\begin{equation*}
	\frac{C(y)}{C(\tilde{y})} \leq \frac{n+2}{n+1}. 
	\end{equation*}
\end{theorem}
\begin{proof}
    We assume without loss of generality (by possibly renaming the agents at each stage) that $x_i^t \leq x_{i+1}^{t}$ for all $i \in [n-1]$ and all $t \in [T]$ so that $x_{\lceil n/2 \rceil}^t$ is the location of the middle agent for each Stage $t$.
    
    Note first that by Lemma \ref{lem:inthemedian}, for each Stage $t \in [T]$, it holds that $\tilde{y}^t \in M^t(\tilde{y}^{t-1})$ and by definition of $A$ it also holds that $y \in M^t(y^{t-1})$. At any stage, define the \emph{non-median cost} of a solution $z \in \mathbb{R}^T$ as 
    \begin{equation*}
    C_{NM}^t(z) = \sum_{i \in [n]\setminus \{\lceil n/2 \rceil\}} |x_i^t - z^t|,
    \end{equation*}
    and define the \emph{residual cost} of $z$ as
    \begin{equation*}
    C_{R}^t(z) = |x_{\lceil n/2 \rceil}^t - z^t| + |z^{t-1} - z^t|,
    \end{equation*}
    so that $C^t(z) = C_{NM}^t(z) + C_R^t(z)$ for all $t \in T$ and for all solutions $z$.
    
    The non-median cost can alternatively be written as
    \begin{equation*}
    C_{NM}^t(z) = \sum_{i = 1}^{\lfloor n/2 \rfloor} (|x_i^t - z^t| + |x_{n-i+1}^t - z^t|),
    \end{equation*}
    and from the latter expression it can be seen that $C_{NM}^t$ is minimised and constant when $z^t$ is in the interval $S^t = [x_{\lceil n/2 \rceil - 1}^t, x_{\lceil n/2 \rceil +1}^t]$, which we refer to as the \emph{supermedian} at Stage $t$. The interval $M^t(z^{t-1})$ is always a subset of the supermedian at Stage $t$, regardless of its argument $z^{t-1}$. Hence, solutions $y$ and $\tilde{y}$ achieve the same non-median cost at any stage. 
    
    Given that both $y$ and $\tilde{y}$ place the facility in the median at every Stage $t$, the facility is placed between $x_{\lceil n/2 \rceil}^t$ and the facility location of the previous stage. Therefore, the residual costs for both solutions at any Stage $t \in [T]$ can be written as $C_R^t(y) = |x_{\lceil n/2 \rceil}^t - y^{t-1}|$ and $C_R^t(\tilde{y}) = |x_{\lceil n/2 \rceil}^t - \tilde{y}^{t-1}|$.
    Thus, we may derive that
    \begin{equation}\label{eq:res}
    C(y) - C(\tilde{y}) = \sum_{t = 1}^T (C_R^t(y) - C_R^t(\tilde{y})) = \sum_{t = 1}^T (|x_{\lceil n/2 \rceil}^t - y^{t-1}| - |x_{\lceil n/2 \rceil}^t - \tilde{y}^{t-1}|).
    \end{equation}
    The remainder of this proof will therefore focus on bouding the right hand side of this equation.
    
    Our approach will be as follows. We classify for each stage the behaviour of the mechanism into one of three \emph{types}. A \emph{type 1 stage} is a stage $t \in [T]$ such that $y^{t-1}$ differs from $\tilde{y}^{t-1}$ and lie on opposing sides of $x_{\lceil n/2 \rceil}^t$. Stage $t$ is a \emph{type 2 stage} if $y^{t-1}$ and $\tilde{y}^{t-1}$ lie on the same side of $x_{\lceil n/2 \rceil}^t$, and the middle point of $M^t(\tilde{y}^{t-1})$ is in $M^t(y^{t-1})$. Lastly, $t$ is a \emph{type 3 stage} if $y^{t-1}$ and $\tilde{y}^{t-1}$ lie on the same side of $x_{\lceil n/2 \rceil}^t$, and the middle point of $M^t(\tilde{y}^{t-1})$ is not in $M^t(y^{t-1})$ (implying that $\tilde{y}^{t-1}$ is further away from $x_{\lceil n/2 \rceil}^t$ than $y^{t-1}$). Note that each state is classified into exactly one of the three types. See Figure \ref{fig:3} for a visualisation of the three types.
    
    
    Our goal is for each stage $t$ to provide a useful bound on the difference in distance by which the facilities at Stage $t$ are removed from the middle agent location $x_{\lceil n/2 \rceil}^{t+1}$, because this difference defines the difference in the residual costs at Stage $t+1$. We bound this difference in distances in terms of the length of the supermedian of Stage $t$ and the optimal residual cost at Stage $t$, i.e., $|\tilde{y}^{t} - x_{\lceil n/2 \rceil}^t|$. Furthermore, we will relate the length of the supermedian at Stage $t$ to the total non-median cost of Stage $t$, which will then yield the desired bound on the competitive ratio.
    
    For a Stage $t \in [T]$, let $\ell^t = |x_{\lceil n/2 \rceil - 1}^t - x_{\lceil n/2 \rceil + 1}^t|$ be the length of the supermedian $S^t$ of Stage $t$.
    For each stage type we provide in separate propositions a meaningful bound. We start with a bound for stages $t$ of type 2, indeed given in terms of the length of the $\ell^t$ of Stage $t$ and the $C_R^t(\tilde{y})$. 
    \begin{proposition}\label{prop:type2}
        For a type $2$ stage $t \in [T]$ it holds that $|y^{t} - \tilde{y}^{t}| \leq \frac{n-1}{2(n+1)} \ell^t + \frac{1}{n+1} C_R^t(\tilde{y})$.
    \end{proposition}
    \begin{proof}
        From the definition of a type $2$ stage it follows that the mechanism places the facility at Stage $t$ exactly at the middle point of $M^t(\tilde{y}^{t-1})$. The distance between $y^t$ and $\tilde{y}^t$ is thus at most half of the length of $M^t(\tilde{y}^{t-1})$, and the length of $M^t(\tilde{y}^{t-1})$ is at most $C_R^t(\tilde{y})$, so $|y^t - \tilde{y}^t| \leq \frac12 C_R^t(\tilde{y})$. The length of $M^t(\tilde{y}^{t-1})$ is also at most the length $\ell^t$ of the supermedian $S^t$, because $M^t(\tilde{y}^{t-1}) \subseteq S^t$. Therefore $|y^t - \tilde{y}^t| \leq \frac12 \ell^t$. Taking a convex combination of these two bounds on $|y^t - \tilde{y}^t|$, we obtain the desired bound 
        \begin{equation*}
        |y^{t} - \tilde{y}^{t}| \leq  \frac12\frac{n-1}{n+1}\ell^t + \frac12\left(1 - \frac{n-1}{n+1}\right) C_R^t(\tilde{y}).
        \end{equation*}
    \end{proof}
    Next, we turn to the type 3 stages. For a type $3$ stage $t$, Mechanism $A$ actually yields a solution $y$ with a \emph{better} cost at Stage $t$ than the globally optimal solution $\tilde{y}$. We prove that the distance $|y^t - \tilde{y}^t|$ between the facilities is at most equal to this profit.
    \begin{proposition}\label{prop:type3}
        For a type $3$ stage $t \in [T]$, it holds that $|y^{t} - \tilde{y}^{t}| \leq C_R^t(\tilde{y}) - C_R^t(y)$.
    \end{proposition}
    \begin{proof}
        From the definition of a type $3$ stage and the mechanism, it follows that $y^{t} = y^{t-1}$. Assume without loss of generality that $x_{\lceil n/2 \rceil}^t$ lies to the right of $y^{t-1} = y^t$. We distinguish two cases.
        \begin{itemize}
        	\item If $\tilde{y}^t$ lies to the left of $y^{t}$, then the distance $|y^t - \tilde{y}^t|$ between the facilities at Stage $t$ is less than the distance $|y^{t-1} - \tilde{y}^{t-1}|$ between the facilities at Stage $t-1$, and the latter defines the difference in residual costs $C_R^t(\tilde{y}) - C_R^t(y)$. 
        	\item The other case is when $\tilde{y}^t$ gets placed to the right of $y^t$. Note that $\tilde{y}^t$ is then in the interval $[y^{t-1},x_{\lceil n/2\rceil}]$, which is shorter than the interval $M^t(\tilde{y}^{t-1})\setminus M^t(y^{t-1})$, which is in turn shorter than the interval $[\tilde{y}^{t-1},y^{t-1}]$. The length of the latter interval defines the difference in residual costs $C_R^t(\tilde{y}) - C_R^t(y)$, which also settles the claim for this second case.
        \end{itemize}
    \end{proof}

Lastly, for a type $1$ stage $t$, we do not prove a direct bound on the distance between the facilities. Rather, the following lemma essentially shows that adding the residual cost difference $C_R^{t+1}(y) - C_R^{t+1}(y)$ of the next stage  to the residual cost difference $C_R^{t}(y) - C_R^{t}(y)$ of the current stage yields a quantity that is at most the distance $|y^{t-1} - \tilde{y}^{t-1}|$ between the facilities at the previous stage (where one should observe that $C_R^{t+1}(y) - C_R^{t+1}(y) \geq |y^{t} - \tilde{y}^t|$ in order to understand this interpretation of the lemma's statement).
\begin{proposition}\label{prop:type1}
    For a type $1$ stage $t$, where $t \in [T-1]$, it holds that 
    \begin{equation*}
    |y^{t} - \tilde{y}^{t}| + \max\{0,C_R^{t}(y) - C_R^{t}(\tilde{y})\} \leq |y^{t-1} - \tilde{y}^{t-1}|.
    \end{equation*}
\end{proposition}
\begin{proof}
    The quantity $\max\{0,C_R^{t}(y) - C_R^{t}(\tilde{y})\}$ is at most $C_R^t(y)$, which is equal to the distance $d_1 := |y^{t-1} - x_{\lceil  n/2 \rceil}^t|$ since $y^t \in M^t(y^{t-1})$. Let $d_2 := |x_{\lceil n/2 \rceil}^t - \tilde{y}^{t-1}|$. Because $\tilde{y}^{t-1}$ and $y^{t-1}$ lie on opposite sides of $x_{\lceil n/2 \rceil}^t$, it holds that $d_1 + d_2 = |y^{t-1} - y^t|$. Mechanism $A$ places the facility at Stage $t$ at the point $y^t = x_{\lceil n/2 \rceil}$. The optimum location $\tilde{y}^t$ at Stage $t$ lies in the interval $M^t(\tilde{y}^{t-1})$ which is a subset of $[x_{\lceil n/2 \rceil}^t, \tilde{y}^{t-1}]$. The distance $|y^t - \tilde{y}^t|$ between the two facilities is therefore at most $d_2$. Therefore,
    \begin{equation*}
    |y^{t} - \tilde{y}^{t}| + \max\{0,C_R^{t}(y) - C_R^{t}(\tilde{y})\} \leq d_2 + C_R^t(y) \leq d_2 + d_1 =  |y^{t-1} - \tilde{y}^{t-1}|. 
    \end{equation*}
\end{proof}

We define a \emph{block} of stages $B \subseteq T$ as a maximal set of subsequent stages $t, \ldots, u$ such that Stages $t$ to $u-1$ are all type $1$ stages (and therefore Stage $u$ is a type 2 or type 3 stage). 
Proposition \ref{prop:type1} implies that for such a block $B = \{t, \ldots, u\}$ the total residual cost difference of block $B$ is bounded by the distance between the facility locations $|y^{t-1} - \tilde{y}^{t-1}|$ in the previous stage $t-1$.
\begin{corollary}\label{cor:type1}
    Let $B = \{t,\ldots, u\} \subseteq [T]$ be a block. Then,
    \begin{equation}\label{eq:block}
    \sum_{s = t}^u \max\{0,C_R^s(y) - C_R^s(\tilde{y})\} \leq |y^{t-1} - \tilde{y}^{t-1}|.
    \end{equation}
\end{corollary}
\begin{proof}
    If $B$ is a singleton then (\ref{eq:block}) states a trivial bound that follows from the definition of $C_R$ and the fact that the facility at stage $t$ gets placed in between the previous location and $x_{\lceil n/2 \rceil}^t$ under both $y$ and $\tilde{y}$. Assume now as an induction hypothesis that (\ref{eq:block}) holds for blocks of size $K$. We prove next that (\ref{eq:block}) also holds if $|B| = K+1$. By the induction hypothesis, we have that 
    \begin{equation*}
    \sum_{s = t+1}^u \max\{0,C_R^s(y) - C_R^s(\tilde{y})\} \leq |y^{t} - \tilde{y}^{t}|.
    \end{equation*}
    By Proposition \ref{prop:type1} it holds that
    \begin{equation*}
    \max\{0,C_R^{t}(y) - C_R^{t}(\tilde{y})\} \leq |y^{t-1} - \tilde{y}^{t-1}| - |y^{t} - \tilde{y}^{t}|,
    \end{equation*}
    hence
    \begin{equation*}
    \sum_{s = t+1}^u \max\{0,C_R^s(y) - C_R^s(\tilde{y})\} =  |y^{t} - \tilde{y}^{t}| + |y^{t-1} - \tilde{y}^{t-1}| - |y^{t} - \tilde{y}^{t}| = |y^{t-1} - \tilde{y}^{t-1}|. 
    \end{equation*}
\end{proof}
   
    Let $\{B_1, \ldots, B_K\}$ be the unique partition of $[T]$ into blocks. We refer to $t_k$ as the final stage of block $k \in [K]$, and for notational convenience we define $t_0 = 0$. Let $T_2$ be the subset of stages $\{t_1,\ldots,t_{K-1}\}$ that are type $2$ stages (which are all type $t$ stages in $[T]$), and let $T_3$ be the subset of $\{t_1, \ldots, t_{K-1}\}$ that are type $3$ stages. 
    
    We can bound the total residual cost difference as follows.
    \begin{eqnarray*}
        \sum_{t = 1}^T (C_R^t(y) - C_R^t(\tilde{y})) & \leq & 
        \sum_{t = 1}^T  \max\{0,C_R^t(y) - C_R^t(\tilde{y})\} + \sum_{t \in T_3}  (C_R^t(y) - C_R^t(\tilde{y})) \\
        & \leq & \sum_{k = 1}^K |y^{t_{k-1}} - \tilde{y}^{t_{k-1}}| + \sum_{t \in T_3}  (C_R^t(y) - C_R^t(\tilde{y})) \\
        & = & \sum_{t \in T_2} |y^{t} - \tilde{y}^{t}| + \sum_{t \in T_3} |y^{t} - \tilde{y}^{t}| | + \sum_{t \in T_3}  (C_R^t(y) - C_R^t(\tilde{y})),
    \end{eqnarray*}
	where the first equality follows from the fact that in type 3 stages, $C_R^t(y) \leq C_R^t(\tilde{y})$, and the first inequality follows from Corollary \ref{cor:type1}.
    We will now apply Propositions \ref{prop:type2} and \ref{prop:type3} to the terms in the last summation. We apply Proposition \ref{prop:type2} to the stages in $T_2$ and we apply Proposition \ref{prop:type3} to the stages in $T_3$, so we obtain:
    \begin{eqnarray*}
        \sum_{t = 1}^T (C_R^t(y) - C_R^t(\tilde{y})) & \leq & \sum_{t \in T_2} |y^{t} - \tilde{y}^{t}| + \sum_{t \in T_3} |y^{t} - \tilde{y}^{t}| + \sum_{t \in T_3}  (C_R^t(y) - C_R^t(\tilde{y})) \\
        & \leq & \sum_{t \in T_2}\left(\frac{n-1}{2(n+1)} \ell^t + \frac{1}{n+1} C_R^t(\tilde{y})\right) \\
        & & \qquad + \sum_{t \in T_3} (C_R^t(\tilde{y}) - C_R^t(y)) \\
        & & \qquad + \sum_{t \in T_3}  (C_R^t(y) - C_R^t(\tilde{y})) \\ 
        & \leq & \sum_{t \in T_2} \left(\frac{n-1}{2(n+1)} \ell^t + \frac{1}{n+1} C_R^t(\tilde{y})\right) \\
        & \leq & \sum_{t = 1}^T  \left(\frac{n-1}{2(n+1)} \ell^t + \frac{1}{n+1} C_R^t(\tilde{y})\right).
    \end{eqnarray*}
    
    
    Under both solutions, the facility is placed in the median at every stage, which is contained in the supermedian. There are $(n-1)/2$ agents left of the supermedian at every stage, and there are $(n-1)/2$ agents right of the supermedian at every stage. Therefore, $((n-1)/2)\ell^t \leq C_{NM}^t(\tilde{y})$ for all $t$. We may therefore bound the last expression in the above derivation as follows.
    \begin{eqnarray*}
        \sum_{t = 1}^T (C_R^t(y) - C_R^t(\tilde{y})) & \leq & \sum_{t = 1}^T  \left(\frac{n-1}{2(n+1)} \ell^t + \frac{1}{n+1} C_R^t(\tilde{y})\right) \\
        & \leq & \sum_{t \in [T]}\left(\frac{1}{n+1}C_{NM}^t(\tilde{y}) + \frac{1}{n+1}C_R^{t}(\tilde{y})\right) \\
        & = & \frac{1}{n+1} \sum_{t \in [T]} C^t(\tilde{y}) = \frac{1}{n+1} C(\tilde{y}).
    \end{eqnarray*}
    From the last bound, we can easily derive the claimed competitive ratio and complete the proof:
    \begin{eqnarray*}
        C(y) & = & C(\tilde{y}) + \sum_{t = 1}^T (C^t(y) - C^t(\tilde{y})) \\
        & = & C(\tilde{y}) + \sum_{t = 1}^T (C_R^t(y) - C_R^t(\tilde{y})) \\
        & \leq & C(\tilde{y}) + \frac{1}{n+1}C(\tilde{y}) \\
        & = & \frac{n+2}{n+1} C(\tilde{y}), 
    \end{eqnarray*}
    where the second equality follows from (\ref{eq:res}).
\end{proof}

\begin{figure}[t]
	\begin{center}
		\includegraphics[scale=1]{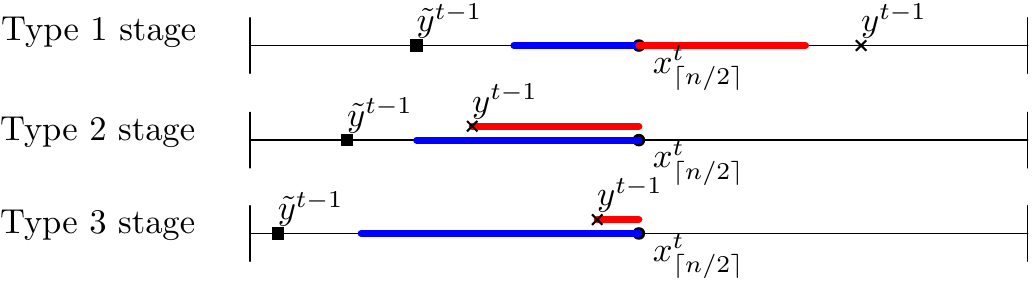}\caption{Depiction of stages belonging to each of the three types. The dot represents the middle agent location $x_{\lceil n/2 \rceil}^t$, the square represents the optimal facility location $\tilde{y}^t$, and the cross represents the facility location $y^{t-1}$ output by the online mechanism. The blue interval represents the median $M^t(\tilde{y}^{t-1})$ associated to the optimal solution, while the red interval represents the median $M^t(y^{t-1})$ associated to the solution output by online mechanism $A$. Agents' locations other than the middle agent location are not depicted. Note that in a type $2$ stage, it may either occur that $y^{t-1}$ is to the left of $\tilde{y}^{t-1}$ or to the right of $\tilde{y}^{t-1}$, although only the latter situation is displayed here.}\label{fig:3}
	\end{center}
\end{figure}

\section{Strategy-proofness}\label{sec:sp}
We investigate in this section the strategy-proofness property of our mechanisms proposed in the previous sections. The results of Moulin \citey{moulin1980strategy} yield a characterisation of the class $\mathcal{C}$ of strategy-proof and group-strategy-proof mechanisms for the classic (single-stage) facility allocation problem: These mechanisms that always place the facility at the median of the union of the set of agent locations and an auxiliary fixed set of points, that are independent of the agent locations.

Unfortunately, the following examples show that the optimal mechanisms for the offline and online settings, which we characterised in Corollary \ref{cor:median2} and Theorem \ref{thm:optonline} respectively, are not group-strategy-proof. This is despite the fact that they can be seen as repeated applications of mechanisms in  $\mathcal{C}$, and this issue can be attributed to the interdependence of the facility locations among the stages.

\begin{example}\label{exa:3}
For even $n$, consider the following instance $I$ with $T=3$ stages and $n=2$ agents, also shown in Figure \ref{fig:4}. The starting location of the facility is $y^0 = 0$. The agent locations are $x^1 = (0, 1)$ and $x^2 = x^3 = (1,0)$. If Agent $1$ does not misreport, their total cost is $2$ under the optimal solution, because the facility will not relocate at all. If Agent $1$ reports instead that she is at location $1$ in in Stage $1$, then the facility will be placed at location $1$ in Stage $1$, and will remain there for the remaining stages, reducing the cost of Agent $1$ by $1$.

For odd $n$, the example is slightly more complex, and is depicted in Figure \ref{fig:5}: Let $T = 4$ and $n=3$. The starting location is $y^0 = 4$. The agent locations are $x^1 = (1,2,5)$, $x^2 = (2,1,4)$, $x^3 = (0,4,5)$, $x^4 = (0,0,0)$. If Agent $1$ does not misreport, then the optimal mechanism of Corollary \ref{cor:median2} outputs the solution $(y^1,y^2,y^3,y^4) = (3,3,3,0)$. This gives agent $1$ a cost $4$. If Agent $1$ instead reports in the first stage that her location is $2$, the vector of facility locations becomes $(2,2,2,0)$, so the cost of Agent $1$ reduces by $1$. Moreover, instance $I'$ demonstrates that Mechanism $A$ of Theorem \ref{thm:optonline} is not strategy-proof: when no agent misreports, the output solution is $(3.5,2.5,3.5,0)$, which yields a total cost of $4.5$ for Agent $1$. If instead, Agent $1$ misreports her location in the first stage as $2$, the output solution becomes $(3,2,3,0)$ which yields Agent $1$ a cost of $3$.
\qed
\end{example}

\begin{figure}[t]
	\begin{center}
		\includegraphics[scale=0.6]{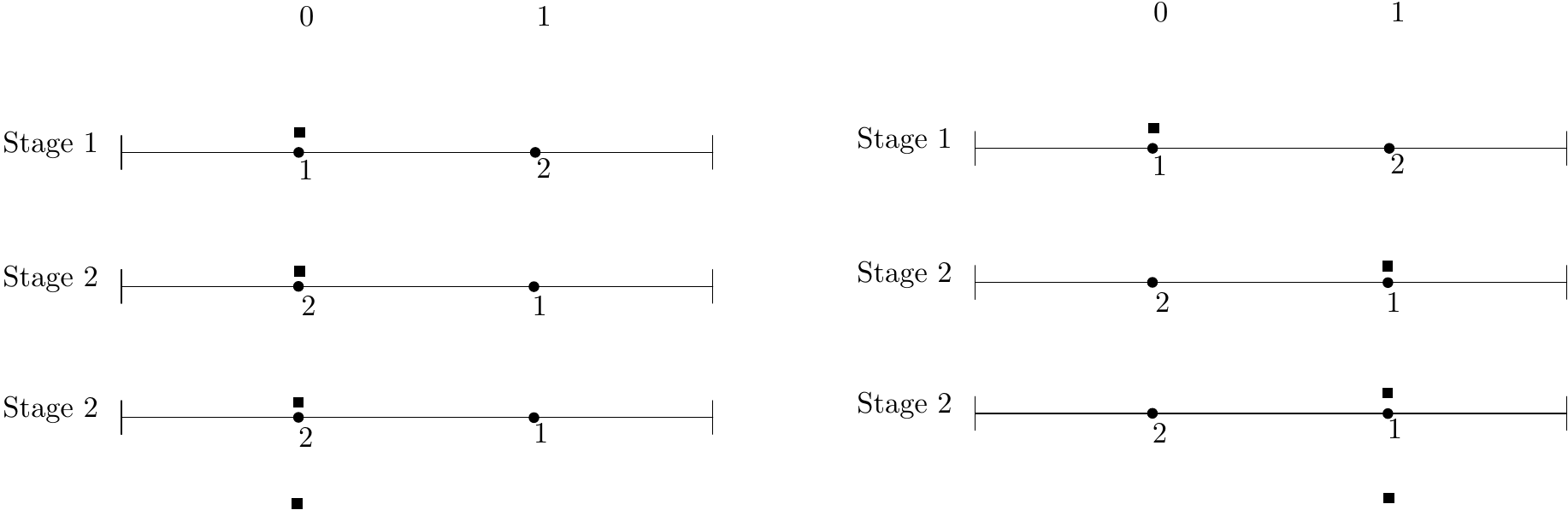}\caption{Depiction of the two agent instance of Example \ref{exa:3}. The notation is identical to previous figures, except that the median interval is not shown, the facility location points are not labeled, and the the agent locations are now labeled with their corresponding agent identities. On the left, the optimal facility locations are depicted when none of the agents misreports. On the right, the optimal facility locations are depicted when Agent $1$ misreports her location in the first stage as $1$.}\label{fig:4}
	\end{center}
\end{figure}

\begin{figure}[t]
	\begin{center}
		\includegraphics[scale=1]{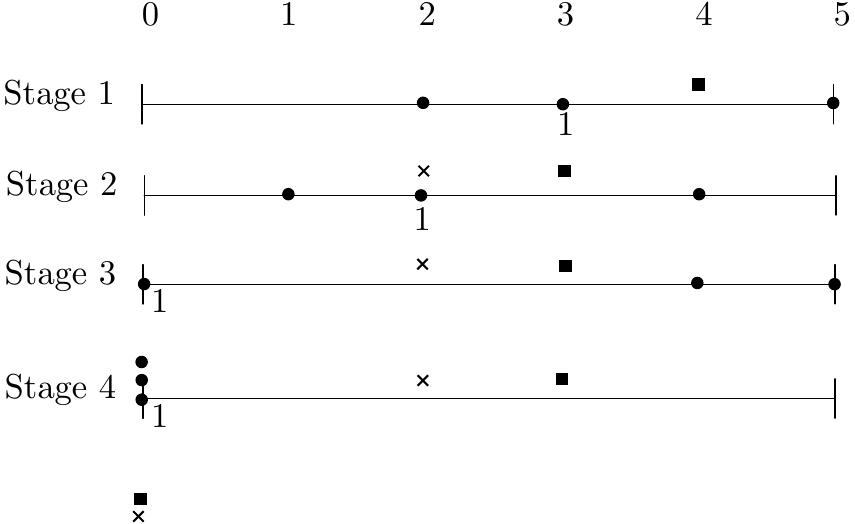}\caption{Depiction of the 3 agent instance of Example \ref{exa:3}. The locations of Agent $1$ are labeled with their identity. The squares depict the optimal facility locations under truthful reporting. Crosses depict the optimal facility locations when Agent $1$ misreports $x_1^1$ as $2$.}\label{fig:5}
	\end{center}
\end{figure}

\begin{corollary}
	The optimal facility reallocation mechanism of Corollary \ref{cor:median2} and the optimal online mechanism of Theorem \ref{thm:optonline} are not strategy-proof. 
\end{corollary}

This establishes that there is a gap between the cost generated by the optimal mechanism and the cost generated by the optimal strategy-proof mechanism. The following simple mechanism bounds this gap. It is an online mechanism performing slightly worse than the optimal online mechanism of Theorem \ref{thm:optonline}, though its competitive ratio still tends to $1$ as the number of agents grows.

\begin{theorem}\label{thm:gspmech}
The online mechanism that puts the facility in every stage at the location of the middle agent (breaking ties arbitrarily in case of even $n$, in a way that is independent of the reported agent locations) is group-strategy-proof and has a competitive ratio of $(n+4)/n$ for even $n$, and $(n+3)/(n+1)$ for odd $n$.
\end{theorem}
\begin{proof}
    Let $(n,T,y^0,x)$ be a facility reallocation instance. Assume without loss of generality that the mechanism always places the facility at the left of the two middle agents, in case $n$ is even. Let $\tilde{y} = (\tilde{y}^1,\ldots,\tilde{y}^T)$ be the optimal solution and let $y = (y^1,\ldots,y^T)$ be the solution output by the mechanism.
    
    We first show strategy-proofness. This follows by straightforward induction on the number of stages. For $T=1$, the instance is a classical facility allocation problem on a line, and our mechanism clearly belongs to Moulin's class $\mathcal{C}$ of group-strategy-proof mechanisms. (In particular, it is easy to see that a misreporting set of agents must fall entirely on one of the two sides of where the facility is placed, but this set of agents can only cause the facility to be placed further away from each of the agents.)
    
    Suppose as an induction hypothesis that all facility reallocation instances of $T-1$ stages are group-strategy-proof. We prove that also every facility reallocation instance of $T$ stages is group-strategy-proof. This game is a sequence of two games: The single stage game $G_1$ consisting of Stage $1$ only, and the $T-1$ stage game $G_2$ consisting of stages $[T]\setminus\{1\}$ with the (reported) middle agent of the first stage as a starting facility location. However, the facility placement at the latter game is independent of the starting location of the facility. Thus, if no agent misreports at Stage $1$, then by the induction hypothesis, no agent is incentivised to misreport at stages $2, \ldots, T$. If a set of agents misreports at Stage $1$, then at least one of the misreporting agents has worse cost at Stage 1, and the induced subgame $G_2$ does not change (as the facility locations in $G_2$ is independent of the starting location). Therefore, group-strategy-proofness holds for $T$ stages. 
    
    Next, we prove the appropriate upper bound on the competitive ratio. To this end, we will assume without loss of generality that $x_i^t \leq x_{i+1}^t$ for all $t \in [T]$ and all $n \in [n-1]$, so that $y^t = x_{\lceil n/2 \rceil}$ by definition of the mechanism.
    
    For odd $n$, we define the supermedian $S^t$ at stage $t$ as $[x_{\lceil n/2 \rceil - 1}^t, x_{\lceil n/2 \rceil + 1}^t]$ (i.e., we use the same definition as in the proof of Theorem \ref{thm:optonline}). For even $n$, we define the supermedian $S^t$ as the interval $[x_{n/2}^t, x_{n/2+1}^t]$.
    
    First, we analyse the difference in distance that the facility moves under both solutions $y$ and $\tilde{y}$. The optimal mechanism places the facility $\tilde{y}^t$ in $M^{t}(\tilde{y}^{t-1})$ for all $t \in [T]$ (by Corollary \ref{cor:median2}) which is contained in $S^t$. Thus, the optimal movement of the facility $|\tilde{y}^{t-1} - \tilde{y}^t|$ at stage $t$ is at least $d(S^t, S^{t-1})$, i.e., the shortest distance between a pair of points in $S^t \times S^{t-1}$, where for convenience we define $S^0 = \{y^0\}$. On the other hand, the facility movement $|y^t - y^{t-1}|$ generated by our mechanism at stage $t$ is at most $\ell^t + \ell^{t-1} + d(S^t,S^{t-1})$, where $\ell^t$ and $\ell^{t-1}$ denote the lengths of $S^t$ and $S^{t-1}$ respectively.
    
    Assume first that $n$ is even.
    It is clear that at every stage, the total distance between the facility location $y^t$ and the agent locations $x_1^t,\ldots,x_n^t$ is the minimum possible, and the same holds for $\tilde{y}^t$. 
    
    This implies that the difference in cost of both solutions is entirely attributed to the difference in total facility movement: 
    \begin{eqnarray*}
    C(y)-C(\tilde{y}) & \leq & \ell^t + \ell^{t-1} + d(S^t, S^{t-1}) - d(S^t, S^{t-1}) \\
    & = &  \sum_{t=1}^T (\ell^t + \ell^{t-1}) = \ell^T + 2 \sum_{t = 1}^{T-1} \ell^t \leq  2\sum_{t = 1}^{T} \ell^t.
    \end{eqnarray*}
    We then make use of the fact that the total distance between any two agents is $\ell^t$, so that the total cost generated by the agents at stage $t$ is at least $(n/2)\ell^t$. The latter implies that 
    \begin{equation*}
    \frac{C(y) - C(\tilde{y})}{C(\tilde{y})} \leq \frac{2\sum_{t = 1}^{T} \ell^t}{(n/2) \sum_{t=1}^T \ell^t} = \frac{4}{n},
    \end{equation*}
    which shows the desired upper bound on the competitive ratio for even $n$. 
    
    Lastly, suppose that $n$ is odd. We first bound the difference between the mechanism and optimum's total facility movement. The mechanism puts the facility at the median agent's location at each stage. Therefore, at each stage, the distance $|y^{t-1} - y^t|$ that the facility moves under the mechanism's output, is at most the distance traveled when the facility is first moved from location $y^{t-1} = x_{\lceil n/2 \rceil}^{t-1}$ to location $\tilde{y}^{t-1}$ as an intermediate step, after which it is moved from $\tilde{y}^{t-1}$ to $\tilde{y}_t$, and finally from $\tilde{y}^t$ to $y^t = x_{\lceil n/2 \rceil}^t$. Therefore, the difference between $y$ and $\tilde{y}$ in total facility movement can be bounded as follows (where for convenience we define $x_{\lceil n/2 \rceil}^0 = y^0 = \tilde{y}^0$).
    \begin{align}
    & \sum_{t=1}^T (|y^{t-1} - y^t| - |\tilde{y}^{t-1} - \tilde{y}^{t}|) \\
    & \qquad \leq \sum_{t = 1}^T (|y^{t-1} - \tilde{y}^{t-1}| + |\tilde{y}^{t-1} - \tilde{y}^t| + |\tilde{y}^t - y^t| - |\tilde{y}^{t-1} - \tilde{y}^{t}|) \notag \\
    & \qquad = \sum_{t = 1}^T |y^{t-1} - \tilde{y}^{t-1}| + |\tilde{y}^t - y^t| \notag \\
    & \qquad = \sum_{t = 1}^T |x_{\lceil n/2 \rceil}^{t-1} - \tilde{y}^{t-1}| + |\tilde{y}^t - x_{\lceil n/2 \rceil}^t| \notag \\
    & \qquad \leq 2 \sum_{t = 1}^T |x_{\lceil n/2 \rceil}^{t} - \tilde{y}^{t}|. \label{eq:movement}
    \end{align}
    The optimum does not always minimise the total distance between the facility and the agents at every stage, although the facility is always placed in $S^t$, so the total distance from the agents to the facility is at each stage $|\tilde{y}^{t} - x_{\lceil n/2 \rceil}^t|$ lower under $y^t$. We subtract this from our bound (\ref{eq:movement}) on the difference in facility movement distance, and we obtain:
    \begin{equation*}
    C(y)-C(\tilde{y}) \leq \sum_{t=1}^T |\tilde{y}^t - x_{\lceil n/2 \rceil}|.
    \end{equation*}
    The quantity $|\tilde{y}^t - x_{\lceil n/2 \rceil}|$ is at most the length $\ell^t$ of the supermedian, so we may bound the above by taking a convex combination of $|\tilde{y}^t - x_{\lceil n/2 \rceil}|$ and $\ell^t$, as follows.
    \begin{eqnarray*}
        C(y)-C(\tilde{y}) & \leq & \sum_{t=1}^T \left(\frac{1}{\lfloor n/2 \rfloor + 1}|\tilde{y}^t - x_{\lceil n/2 \rceil}| + \left(1 - \frac{1}{\lfloor n/2 \rfloor + 1}\right)\ell^t\right) \\ 
    \end{eqnarray*}
    At stage $t$, the distance between any two agents except the middle agent is $\ell^t$, and the distance between the middle agent and the facility is $|\tilde{y}^t - x_{\lfloor n/2 \rfloor}|$. Hence, $|\tilde{y}^t - x_{\lceil n/2 \rceil}| + \lfloor n/2 \rfloor \ell^t $ is a lower bound on $C^t(\tilde{y})$, and we use that to bound the following ratio:
    \begin{eqnarray*}
        \frac{C(y) - C(\tilde{y})}{C(\tilde{y})} & \leq & \frac{\sum_{t=1}^T \left(\frac{1}{\lfloor n/2 \rfloor + 1}|\tilde{y}^t - x_{\lceil n/2 \rceil}| + \left(1 - \frac{1}{\lfloor n/2 \rfloor + 1}\right)\ell^t\right)}{\sum_{t=1}^T (|\tilde{y}^t - x_{\lceil n/2 \rceil}| + \lfloor n/2 \rfloor \ell^t)} \\
        & = & \frac{\frac{1}{\lfloor n/2 \rfloor +1}\sum_{t=1}^T \left(|\tilde{y}^t - x_{\lceil n/2 \rceil}| + \lfloor n/2\rfloor \ell^t\right)}{\sum_{t=1}^T (|\tilde{y}^t - x_{\lceil n/2 \rceil}| + \lfloor n/2 \rfloor \ell^t)} \\
        & = & \frac{1}{\lfloor n/2 \rfloor +1} \\
        & = & \frac{2}{n+1},
    \end{eqnarray*} 
    which shows the desired upper bound for odd $n$.
\end{proof}

\begin{example}\label{exa:4}
The following family of examples shows that the analysis of the competitive ratio in Theorem \ref{thm:gspmech} is tight for all $n$. Let the starting facility location be $y^0 = 1$ and let there be two stages. In Stage $1$, Agents $1$ to $\lfloor n/2 \rfloor$ are located at $1$, and the remaining agents are located at $0$. In Stage $2$, all of the agents are located at $1$. The optimal Mechanism (see Corollary \ref{cor:median2}) places the facility at location $1$ in both stages (regardless of whether $n$ is odd of even), which results in a cost of $n/2$ if $n$ is even, and a cost of $(n+1)/2$ if $n$ is odd.

The mechanism of Theorem \ref{thm:gspmech} places the facility at location $0$ in the first stage, and at location $1$ in the second stage. This yields a total cost of $n/2 + 2$ if $n$ is even, and a total cost of $(n-1)/2 + 2$ when $n$ is odd. Thus, when $n$ is even, the competitive ratio on these instances is $((n/2) + 2)/(n/2) = (n+4)/n$, and when $n$ is odd, the competitive ratio is $((n+3)/2)/(n+1)/2 = (n+3)/(n+1)$.
\qed
\end{example}

\section{Discussion}

We studied a multi-stage variant of the classical facility location problem, where the problem is repeated over multiple stages and there is a cost incurred by moving the facility across stages. 

Our focus in this work has primarily been on identifying and computing the optimal facility placement and movement. We characterised the optimal mechanisms both in the offline and online setting. We considered this problem under the constraint of strategy-proofness as well. These mechanisms turn out to be elegant and simple in their definition, but are surprisingly challenging to analyse. Finally, we showed that neither of these mechanisms is strategy-proof, and devised a new strategy-proof mechanism. We analysed the performance of this strategy-proof mechanism in the online setting. 

Our mechanism definitions, and the properties that we prove about them, reveal some interesting insights, such as the discrepancy between the cases of an even and an odd number of agents, and the fact that there is a single-stage ``lookahead'' needed to achieve optimality in the odd case.

Interesting future directions are to design online and strategy-proof mechanisms for the generalised variant of the problem that we briefly considered, and to characterise the class of (group)-strategy-proof mechanisms for the basic version of the problem. We conjecture that the competitive ratio of Theorem \ref{thm:gspmech} is the best achievable among the strategy-proof mechanisms. Additionally, randomised mechanisms can be studied in this context as it is known that they outperform deterministic ones in the single stage case \cite{lu2010asymptotically}.

An alternative generalisation of the problem that would be interesting (and undoubtedly more complex) to study is to increase the dimension of the Euclidian space in which the locations lie, e.g. to consider facility reallocation on the plane instead of the line.

\section*{Acknowledgments}
The first author was partially supported by 
NWO grant 612.001.352, EPSRC
grant EP/P020909/1, and University of Liverpool Visiting Fellowship.
The second author was partially supported by 
EPSRC grants EP/M027287/1 and EP/P020909/1.
Part of the preparation of this manuscript was while the first author was appointed as a lecturer at University of Essex. 
We would like to thank Orestis Telelis and Guido Sch\"{a}fer for helpful discussion that led to this paper.

%

\bibliographystyle{plain}
\bibliography{reallocationjournal}

\providecommand{\noopsort}[1]{}
\begin{thebibliography}{10}

\bibitem{ahmadian2013local}
S.~Ahmadian, Z.~Friggstad, and C.~Swamy.
\newblock Local-search based approximation algorithms for mobile facility
  location problems.
\newblock In {\em Proceedings of the Twenty-Fourth Annual ACM-SIAM Symposium on
  Discrete Algorithms}, pages 1607--1621. Society for Industrial and Applied
  Mathematics, 2013.

\bibitem{evolving1}
H.-C. An, A.~Norouzi-Fard, and O.~Svensson.
\newblock Dynamic facility location via exponential clocks.
\newblock {\em ACM Transactions on Algorithms}, 13(2):21:1--21:20, 2017.

\bibitem{Chen2020185}
Z.~Chen, K.~C.~K. Fong, M.~Li, K.~Wang, H.~Yuan, and Z.~Yong.
\newblock Facility location games with optional preference.
\newblock {\em Theoretical Computer Science}, 847:185--197, 2020.

\bibitem{diveki2011online}
G.~Div{\'e}ki and C.~Imreh.
\newblock Online facility location with facility movements.
\newblock {\em Central European Journal of Operations Research},
  19(2):191--200, 2011.

\bibitem{downs1957economic}
A.~Downs.
\newblock An economic theory of political action in a democracy.
\newblock {\em Journal of Political Economy}, 65(2):135--150, 1957.

\bibitem{duan2019heterogeneous}
L.~Duan, B.~Li, M.~Li, and X.~Xu.
\newblock Heterogeneous two-facility location games with minimum distance
  requirement.
\newblock In {\em Proceedings of the 18th International Conference on
  Autonomous Agents and Multi-Agent Systems}, pages 1461–--1469.
  International Foundation for Autonomous Agents and Multiagent Systems, 2019.

\bibitem{evolving3}
D.~Eisenstat, C.~Mathieu, and N.~Schabanel.
\newblock Facility location in evolving metrics.
\newblock In {\em Proceedings of the International Colloquium on Automata,
  Languages and Programming}, pages 459--470. Springer, 2014.

\bibitem{feldman2016voting}
M.~Feldman, A.~Fiat, and I.~Golomb.
\newblock On voting and facility location.
\newblock In {\em Proceedings of the 2016 ACM Conference on Economics and
  Computation}, pages 269--286. ACM, 2016.

\bibitem{filos2017facility}
Aris Filos-Ratsikas, Minming Li, Jie Zhang, and Qiang Zhang.
\newblock Facility location with double-peaked preferences.
\newblock {\em Autonomous Agents and Multi-Agent Systems}, 31(6):1209–--1235,
  2017.

\bibitem{fotakis2011online}
D.~Fotakis.
\newblock Online and incremental algorithms for facility location.
\newblock {\em ACM SIGACT News}, 42(1):97--131, 2011.

\bibitem{fotakisetal}
D.~Fotakis, L.~Kavouras, P.~Kostopanagiotis, P.~Lazos, S.~Skoulakis, and
  N.~Zarifis.
\newblock Reallocating multiple facilities on the line.
\newblock {\em Theoretical Computer Science}, 858:13--34, 2021.

\bibitem{fotakis2014power}
D.~Fotakis and C.~Tzamos.
\newblock On the power of deterministic mechanisms for facility location games.
\newblock {\em ACM Transactions on Economics and Computation}, 2(4):15, 2014.

\bibitem{friggstad2011minimizing}
Z.~Friggstad and M.~R. Salavatipour.
\newblock Minimizing movement in mobile facility location problems.
\newblock {\em ACM Transactions on Algorithms}, 7(3):28, 2011.

\bibitem{evolving2}
A.~Gupta, K.~Talwar, and U.~Wider.
\newblock Changing bases: Multistage optimization for matroids and matchings.
\newblock In {\em Proceedings of the International Colloquium on Automata,
  Languages and Programming}, pages 563--575. Springer, 2014.

\bibitem{de2018facility}
B.~{\noopsort{Keijzer}}{de Keijzer} and D.~Wojtczak.
\newblock Facility reallocation on the line.
\newblock In {\em Proceedings of the 27th International Joint Conference on
  Artificial Intelligence}, pages 188--194. {AAAI} Press, 2018.

\bibitem{lin2012learning}
H.~Lin and J.~Bilmes.
\newblock Learning mixtures of submodular shells with application to document
  summarization.
\newblock In {\em Proceedings of the 28th Conference on Uncertainty in
  Artificial Intelligence (UAI'12)}, pages 479--490. AUAI Press, 2012.

\bibitem{lu2010asymptotically}
P.~Lu, X.~Sun, Y.~Wang, and Z.~A. Zhu.
\newblock Asymptotically optimal strategy-proof mechanisms for two-facility
  games.
\newblock In {\em Proceedings of the 11th ACM Conference on Electronic
  Commerce}, pages 315--324. {ACM}, 2010.

\bibitem{lu2009tighter}
P.~Lu, Y.~Wang, and Y.~Zhou.
\newblock Tighter bounds for facility games.
\newblock In {\em Proceedings of the International Workshop on Internet and
  Network Economics}, pages 137--148. Springer, 2009.

\bibitem{megiddo1984complexity}
N.~Megiddo and K.~J. Supowit.
\newblock On the complexity of some common geometric location problems.
\newblock {\em SIAM Journal on Computing}, 13(1):182--196, 1984.

\bibitem{megiddo1983maximum}
N.~Megiddo, E.~Zemel, and S.~L. Hakimi.
\newblock The maximum coverage location problem.
\newblock {\em SIAM Journal on Algebraic Discrete Methods}, 4(2):253--261,
  1983.

\bibitem{miyagawa2001locating}
E.~Miyagawa.
\newblock Locating libraries on a street.
\newblock {\em Social Choice and Welfare}, 18(3):527--541, 2001.

\bibitem{moulin1980strategy}
H.~Moulin.
\newblock On strategy-proofness and single peakedness.
\newblock {\em Public Choice}, 35(4):437--455, 1980.

\bibitem{procaccia2009approximate}
A.~D. Procaccia and M.~Tennenholtz.
\newblock Approximate mechanism design without money.
\newblock In {\em Proceedings of the 10th ACM Conference on Electronic
  Commerce}, pages 177--186. ACM, 2009.

\bibitem{procacciaapproximation}
A.~D. Procaccia, D.~Wajc, and H.~Zhang.
\newblock Approximation-variance tradeoffs in facility location games.
\newblock In {\em Proceedings of the Thirty-Second {AAAI} Conference on
  Artificial Intelligence}. {AAAI} Press, 2017.

\bibitem{serafino2015truthful}
P.~Serafino and C.~Ventre.
\newblock Truthful mechanisms without money for non-utilitarian heterogeneous
  facility location.
\newblock In {\em Proceedings of the Twenty-Ninth AAAI Conference on Artificial
  Intelligence}, pages 1029--1035. {AAAI} Press, 2015.

\bibitem{songCVPR17}
H.~O. Song, S.~Jegelka, R.~Vivek, and M.~Kevin.
\newblock Deep metric learning via facility location.
\newblock In {\em 2017 IEEE Conference on Computer Vision and Pattern
  Recognition (CVPR)}, pages 2206--2214. {IEEE}, 2017.

\bibitem{sui2013analysis}
X.~Sui, C.~Boutilier, and T.~Sandholm.
\newblock Analysis and optimization of multi-dimensional percentile mechanisms.
\newblock In {\em Proceedings of the Twenty-Third International Joint
  Conference on Artificial Intelligence}, pages 367--374. {AAAI} Press, 2013.

\bibitem{todo2011false}
T.~Todo, A.~Iwasaki, and M.~Yokoo.
\newblock False-name-proof mechanism design without money.
\newblock In {\em The 10th International Conference on Autonomous Agents and
  Multiagent Systems - Volume 2}, pages 651--658. International Foundation for
  Autonomous Agents and Multiagent Systems, 2011.

\bibitem{tschiatschek2014learning}
S.~Tschiatschek, R.~K. Iyer, H.~Wei, and J.~A. Bilmes.
\newblock Learning mixtures of submodular functions for image collection
  summarization.
\newblock In {\em Advances in Neural Information Processing Systems}, pages
  1413--1421. MIT Press, 2014.

\bibitem{weber1909standort}
A.~Weber.
\newblock {\"U}ber den standort der industrien, 1. teil: Reine theorie des
  standortes. (on the location of industries), 1909.

\bibitem{xu2021two}
X.~Xu, B.~Li, M.~Li, and L.~Duan.
\newblock Two-facility location games with minimum distance requirement.
\newblock {\em Journal of Artificial Intelligence Research}, 70:719--756, 2021.

\bibitem{zou2015facility}
S.~Zou and M.~Li.
\newblock Facility location games with dual preference.
\newblock In {\em Proceedings of the 2015 International Conference on
  Autonomous Agents and Multiagent Systems}, pages 615--623. International
  Foundation for Autonomous Agents and Multiagent Systems, 2015.

\end{thebibliography}


\end{document}